\documentclass[preprint,authoryear,12pt]{elsarticle}

\RequirePackage[colorlinks,citecolor=blue,urlcolor=blue]{hyperref}

\usepackage{a4wide}
\usepackage{amssymb,amsmath, amsthm}
\usepackage{ mathrsfs }
\usepackage{graphicx}
\usepackage{color}
\usepackage{bbm} 
\usepackage{framed}
\usepackage{enumerate}
\usepackage{multirow}
\usepackage{natbib}

\newtheorem{thm}{Theorem}

\newtheorem{lem}[thm]{Lemma}
\theoremstyle{definition}

\newtheorem*{move}{Move}

\providecommand{\NN}{\mathbb{N}}

\providecommand{\RR}{\mathbb{R}}

\providecommand{\mc}[1]{\mathcal{#1}}

\providecommand{\mr}[1]{\mathrm{#1}}

\providecommand{\ind}[1]{{1}_{#1}}

\providecommand{\diag}{{\operatorname{diag}}}

\newcommand{\dd}{{\,\mathrm d}}

\let\tilde\widetilde

\renewcommand{\phi}{\varphi}
\newcommand{\given}{\,|\,}

\renewcommand{\th}{\theta}

\newcommand{\scale}{{s^2}}

\newcommand\bquot{B(j'\given j)}
\newcommand\rquot{R(j'\given j)}

\journal{Computational Statistics and Data Analysis}

\begin{document}

\begin{frontmatter}


\title{Reversible jump MCMC for nonparametric drift estimation for diffusion processes}

\author[diam]{Frank van der Meulen}
\author[diam]{Moritz Schauer\corref{corresp}}
\author[kdv]{Harry van Zanten}
\cortext[corresp]{Corresponding author. Address: TU Delft, Mekelweg 4, 2628 CD Delft, The Netherlands.  E-mail: {\tt m.r.schauer@tudelft.nl}. Tel: ++31 15 2782546.}
\address[diam]{Delft Institute for Applied Mathematics (Delft University of Technology).}
\address[kdv]{Korteweg-de Vries Institute for Mathematics (University of Amsterdam).}



\begin{abstract}
In the context of nonparametric Bayesian estimation a Markov chain Monte Carlo algorithm is devised and implemented to sample from the posterior distribution of the drift function of a continuously or discretely observed one-dimensional diffusion. The drift is modeled by a scaled linear combination of basis functions with a Gaussian prior on the coefficients. The scaling parameter is equipped with a partially conjugate prior.  The number of basis function in the drift is equipped with a prior distribution as well.  For continuous data, a reversible jump Markov chain algorithm enables the exploration of the posterior over  models of varying dimension. Subsequently, it is explained  how data-augmentation can be used to extend  the algorithm to deal with diffusions observed discretely in time. Some examples illustrate that the method can give satisfactory results. In these examples a comparison is made with another existing method as well.
\end{abstract}



\end{frontmatter}



\section{Introduction}
Suppose we observe a diffusion process $X$, given as the  solution of the stochastic differential equation (SDE)
\begin{equation}\label{eq:sde}
 \dd X_t = b(X_t) \dd t + \dd W_t,\qquad  X_0 = x_0,
\end{equation}
with initial state $x_0$ and unknown drift function $b$. The aim is to estimate the drift $b$ when a sample path of the diffusion is observed continuously up till
a time $T > 0$  or at discrete times $0, \Delta, {2\Delta}, \ldots, {n\Delta}$, for some $\Delta > 0$ and $n \in \NN$.

Diffusion models are widely employed in a variety of scientific fields, including physics, economics and biology.
Developing methodology for fitting SDEs  to observed data has therefore become an important problem.
In this paper we restrict the exposition to the case that the drift function is 1-periodic and the diffusion function is identically equal to $1$.
This is motivated by applications in which the data consists of recordings of angles, cf.\ e.g.\ \cite{Yvo}, \cite{Hindriks} or \cite{Pokern}.
The methods we propose can however be adapted to work in more general setups, such as ergodic diffusions with non-unit diffusion coefficients.
In the continuous observations case, a diffusion with periodic drift could alternatively be viewed as diffusion on the circle. Given only discrete observations on the circle, the information about how many turns around the circle the process has made between the observations is lost however and the total number of  windings is unknown. For a lean exposition we concentrate therefore on diffusions with periodic drift on $\RR$. In the discrete observations setting the true circle case could be treated by introducing a latent variable that keeps track of the winding number.

In this paper we propose a new approach to making nonparametric Bayesian inference for the model (\ref{eq:sde}).
A Bayesian method can be attractive since it does not only yield an estimator for the unknown drift function, but also gives
a quantification of the associated uncertainty through the spread of the posterior distribution, visualized for instance
by pointwise credible intervals.
Until now the development of Bayesian methods for diffusions has largely focussed on parametric models. In such models
it is assumed that the drift is known up to a finite-dimensional parameter  and the problem reduces to making inference
about that parameter. See for instance the papers \cite{Eraker}, \cite{RobertsStramer}, \cite{BeskosPapaspiliopoulosRoberts},
 to mention but a few.
When no obvious parametric specification of the drift function is available
it is sensible to explore  nonparametric estimation methods, in order to reduce the risk  of model misspecification
or to validate certain parametric specifications.
The literature on nonparametric Bayesian methods for SDEs is however still very limited at  the present time. The only
paper which proposes a practical method we are aware of is \cite{Pokern}. The theoretical, asymptotic behavior of the procedure
of \cite{Pokern} is studied in the recent paper \cite{PokernStuartvanZanten}. Other papers dealing with asymptotics in this framework include
\cite{PanzarVZanten} and \cite{VdMeulenVZanten}, but these do not propose practical computational methods.

The approach we develop in this paper extends or modifies that of \cite{Pokern} in a number of directions
and employs different numerical methods.
\cite{Pokern} consider  a  Gaussian prior distribution on the periodic drift function $b$. This prior is defined
as a Gaussian distribution on $L^2[0,1]$ with densely defined inverse covariance operator
(precision operator)
\begin{equation}\label{eq: *}
\eta ((-\Delta)^p + \kappa I),
\end{equation}
where $\Delta$ is the one-dimensional Laplacian (with periodic boundaries conditions), $I$ is the identity operator and
$\eta, \kappa > 0$ and  $p \in \NN$ are fixed hyper parameters.
 It is asserted in \cite{Pokern} and proved in \cite{PokernStuartvanZanten} that
if the diffusion is observed continuously, then for this prior the posterior mean can be characterized as the weak solution of a certain
differential equation involving the local time of the diffusion. Moreover,  the  posterior precision operator can be
explicitly expressed as a differential operator as well. Posterior
computations can then be done using numerical methods for differential equations.

To explain our alternative  approach we note, as in \cite{PokernStuartvanZanten}, that the prior  just defined
can be described equivalently  in terms of series expansions.  Define the basis functions $\psi_k \in L^2[0,1]$ by setting $\psi_1\equiv 1$, and for $k \in \NN $
$\psi_{2k}(x) = \sqrt{2}\sin(2k\pi x)$ and $\psi_{2k+1}(x) = \sqrt{2}\cos(2k\pi x)$.
Then the prior  is the law of the random function
\[
x \mapsto \sum_{l=1}^\infty \sqrt{\lambda_l }Z_l \psi_l(x),
\]
where the $Z_l$ are independent, standard normal variables and for $l \ge 2$
\begin{equation}\label{eq:prior:pokern}
\lambda_l = \Big(\eta \Big(4\pi^2\Big\lceil\frac{l}{2}\Big\rceil^2\Big)^p + \eta\kappa\Big)^{-1}.
\end{equation}
This characterization shows in particular that the hyper parameter $p$ describes the regularity of the prior
through the decay of the Fourier coefficients
and $1/\eta$ is a multiplicative scaling parameter.
The priors we consider in this paper are also defined via series expansions. However, we make a number
of substantial changes.

Firstly, we allow for different types of basis functions.
Different basis functions instead of the Fourier-type functions may be computationally attractive.
The posterior computations involve the inversion of certain large matrices and choosing basis functions with
local support  typically makes these matrices sparse.
In the general exposition we keep the basis
functions completely general but  in the simulation results we will consider wavelet-type Faber--Schauder functions  in addition
to the Fourier basis. A second difference is that we truncate the infinite series at a level that we endow with a prior as well.
In this manner we can achieve considerable computational gains if the data driven truncation point is relatively small,
so that only low-dimensional models are used and hence only relatively small matrices have to be inverted.
A last important change is that we do not set the multiplicative hyper parameter at a fixed value, but instead
endow it with a prior and let the data determine the appropriate value.

We will present simulation results in Section \ref{sec: sim} which illustrate that our approach indeed has several
advantages. Although the truncation of the series at a data driven point
involves incorporating reversible jump MCMC steps in our computational algorithm,
we will show that it can indeed lead to a considerably faster procedure compared to truncating
at some fixed high level. The introduction of a prior on the multiplicative hyper parameter reduces
the risk of misspecifying the scale of the drift. We will  show in Section \ref{sec:scale} that using a fixed scaling parameter
can seriously deteriorate the quality of the inference, whereas our hierarchical procedure with a prior on that parameter
is able  to adapt to the true scale of the drift.
A last advantage that we will illustrate numerically is that by introducing both  a prior on the scale and on the truncation level we
 can achieve some degree of adaptation to smoothness as well.

Computationally we use a combination of methods that are well established in other statistical settings.
Within models in which the truncation point of the series is fixed we use  Gibbs sampling based
on standard inverse gamma-normal computations. We combine this with reversible jump MCMC to
move between different models. For these moves, we can use an auxiliary Markov chain to propose a model, and subsequently draw  coefficients from their posterior distribution within that model. Such a scheme has been proposed for example in \cite{Godsill} for estimation in autoregressive time-series models.  In  case of discrete observations we also incorporate a data augmentation
step using a Metropolis--Hastings sampler to generate diffusion bridges. Our numerical examples illustrate that
using our algorithm it  is computationally feasible to carry out nonparametric Bayesian inference for low-frequency diffusion data
using a non-Gaussian hierarchical prior which is more flexible than previous methods.

A brief outline of the article is as follows: In Section \ref{sec:prior} we give a concise prior specification. In the section thereafter, we present the reversible jump algorithm to draw from the posterior for continuous-time data. Data-augmentation is discussed in Section  \ref{dataaug}. In Section \ref{sec: sim} we give some examples to illustrate our method. We end with a section on numerical details.

\section{Prior distribution}\label{sec:prior}

\subsection{General prior specification}

To define our prior on the periodic drift function $b$ we write a truncated series expansion for $b$
and put prior weights on the truncation point and on the coefficients in the expansion.
We employ general
$1$-periodic, continuous basis functions $\psi_l$, $l \in \NN$. In the concrete examples
ahead we will consider in particular Fourier and Faber-Schauder functions. We fix an increasing sequence
of natural numbers $m_j$, $j \in \NN$, to group the basis functions into {\em levels}. The
functions $\psi_1, \ldots, \psi_{m_1}$ constitute level $1$, the functions $\psi_{m_1+1}, \ldots, \psi_{m_2}$ correspond to
level $2$, etcetera. In this manner we can accommodate both families of basis functions with a single index (e.g.\ the Fourier basis)
and doubly indexed families (e.g.\ wavelet-type bases) in our framework. Functions that are linear combinations
of the first $m_j$ basis functions $\psi_1, \ldots, \psi_{m_j}$ are said to belong to {\em model} $j$. Model $j$ encompasses levels $1$ up till $j$.

To define the prior on $b$ we first put a prior on the model index $j$, given by certain prior weights $p(j)$, $j \in \NN$.
By construction, a function in model $j$ can be expanded as
$\sum_{l=1}^{m_j} \theta^j_l\psi_l$
for a certain vector of coefficients $\theta^j \in \RR^{m_j}$. Given $j$, we endow this vector with a prior by postulating that
the coefficients $\theta^j_l$ are given by an inverse gamma scaling constant times independent, centered Gaussians
with decreasing variances $\xi^2_l$, $l \in \NN$. The  choice of the constants $\xi^2_l$ is discussed in Sections \ref{sec: fourier} and \ref{sec: schauder}.

Concretely, to define the prior we fix model probabilities $p(j)$, $j \in \NN$, decreasing variances $\xi^2_l$,
positive constants $a, b > 0$ and set $\Xi^j = \diag({\xi^2_1, \ldots, \xi^2_{m_j}})$.
Then the hierarchical prior $\Pi$ on the drift function $b$ is defined as follows:

\medskip

\begin{center}
\framebox{
\parbox{0.6\textwidth}{
\begin{align*}
j & \sim p(j),\\
s^2 & \sim {\rm IG}(a,b),\\
\theta^j\given j, s^2 & \sim N_{m_j}(0, s^2\Xi^j),\\
b \given j, s^2, \theta^j & \sim \sum_{l=1}^{m_j} \theta^j_l\psi_l.
\end{align*}

}}
\end{center}

\subsection{Specific  basis functions}

Our general setup is chosen such that we can incorporate bases indexed by a single number and doubly indexed (wavelet-type) bases.
For the purpose of illustration, one example for each case is given below. First, a Fourier basis expansion, which emphasizes the (spectral) properties of the drift in frequency domain and second, a (Faber--) Schauder system which features basis elements with local support.

\subsubsection{Fourier basis}
\label{sec: fourier}

In this case we set  $m_j = 2j-1$ and the basis functions are defined as
\[
\psi_{1}\equiv 1, \quad \psi_{2k}(x) = \sqrt2\sin(2k\pi x), \quad \psi_{2k+1}(x) = \sqrt2\cos(2k\pi x),\quad k \in \NN.
\]
These functions form an orthonormal basis of $L^2[0,1]$ and the decay of the Fourier coefficients of a function is
related to its regularity. More precisely, if $f = \sum_{l \ge 1} \theta_l \psi_l$ and $\sum_{l \ge 1} \theta^2_l l^{2\beta} < \infty$
for $\beta > 0$, then $f$ has Sobolev regularly $\beta$, i.e.\ it has square integrable weak derivatives up to the order $\beta$.
By setting $\xi^2_l \sim l^{-1-2\beta}$ for $\beta > 0$,  we  obtain a prior  which  has a version with $\alpha$-H\"older continuous sample paths for all $\alpha<\beta$.
A possible choice for the model probabilities is to take them geometric, i.e.\ $p(j) \sim \exp(-Cm_j)$ for some $C > 0$.

Priors of this type are quite common in other statistical settings. See for instance  \cite{Zhao} and \cite{ShenWasserman},
who considered priors of this type in the context of the white noise model and nonparametric regression. The prior
can be viewed as an extension of the one
 of \cite{Pokern} discussed in the introduction.
The latter uses the same basis functions and decreasing variances with $\beta = p-1/2$. It does not put a prior on the model index $j$ however
(it basically takes $j = \infty$) and uses a fixed scaling parameter whereas we put a prior on $s$.
In Section \ref{sec: sim} we argue that our approach has a number of advantages.

\subsubsection{Schauder functions}
\label{sec: schauder}

The Schauder basis functions are a location and scale family based on the ``hat'' function
$\Lambda(x) = (2x)\ind{[0,\frac12)}(x) + 2(x-1) \ind{[\frac12,1]}(x)$.
With $m_j = 2^{j-1}$, the Schauder system is given by $\psi_{1} \equiv 1$ and for $l \ge 2$ $\psi_l(x) = \Lambda_l(x \operatorname{mod} 1)$, where
$$\Lambda_{2^{j-1} + k}(x) = \Lambda(2^{j-1}x - k+1), \quad j\ge 1,  \quad k=1,\ldots, 2^{j-1}.$$
These functions have compact supports. A Schauder expansion thus emphasizes local properties of the sample paths.
For $\beta \in (0,1)$, a function $f$ with Faber--Schauder expansion $f = \sum_{l \ge 1} c_l\psi_l= c_1\psi_1 + \sum_{j\ge 1}\sum_{k=1}^{2^{j-1}} c_{2^{j-1}+k} \psi_{2^{j-1}+k}$ has H\"older regularity
of order $\beta$ if and only if  $|c_l| \le  {\rm const.} \times l^{-\beta}$ for every $l$ (see for instance
\cite{kashin}). It follows that if in our setup we take
$\xi_{2^{j-1} + k} = 2^{-\beta j}$ for  $j\ge 1$ and $k=1,\ldots, 2^{j-1} $, then we obtain a prior with regularity $\beta$.
A natural choice for $p(j)$ is again
$ p(j) \sim \exp(-C m_j)$.

The Schauder system is well known in the context of constructions of  Brownian motion,
see for instance \cite{RW}.
The Brownian motion case corresponds to prior regularity   $\beta = 1/2$.


\section{The Markov chain Monte Carlo sampler}

\subsection{Posterior within a fixed model}\label{likelihood}\label{sec: within}
When continuous observations $x^T = (x_t: t \in [0,T])$ from the diffusion model (\ref{eq:sde}) are available,
then we have an explicit expression for the likelihood $p(x^T \given b)$. Indeed, by Girsanov's formula
we almost surely have
\begin{equation}\label{eq: girsanov}
p(x^T \given b) =  \exp\Big( \int_0^T b(x_t) \dd x_t- \frac12 \int_0^T b^2(x_t) \dd t\Big) .
\end{equation}
Cf.\ e.g. \cite{LiptserShiryayevI}. Note in particular that the log-likelihood is quadratic in $b$.

Due to the special choices in the construction of our hierarchical prior, the quadratic structure  implies that
within a fixed model $j$, we can do  partly explicit posterior computations.
More precisely, we can derive the posterior distributions of the scaling constant $s^2$
and the vector of coefficients $\theta^j$ conditional on all the other parameters. The continuous observations
enter the expressions through the vector $\mu^j \in \RR^{m_j}$ and the $m_j\times m_j$ matrix $\Sigma^j$ defined by
\begin{equation}\label{eq:mu}	
\mu^j_l = \int_0^T \psi_l(x_t) \dd x_t, \quad l = 1, \ldots, m_j,
\end{equation}
and
\begin{equation}\label{eq:sigma}
\Sigma^j_{l, l'} = \int_0^T \psi_l(x_t)\psi_{l'}(x_t) \dd t,  \quad l, l' = 1, \ldots, m_j.
\end{equation}

\bigskip

\begin{lem}\label{lem:postcoeff}
We have
\begin{align*}
\theta^j & \given  s^2, j, x^T \sim N_{m_j}((W^j)^{-1}\mu^j, (W^j)^{-1}),\\
s^2 & \given \theta^j, j, x^T \sim {\rm IG}(a + (1/2) m_j, b + (1/2) (\th^{j})^T (\Xi^j)^{-1} \th^j ),
\end{align*}
where $W^j= \Sigma^j +  (s^2\Xi^j)^{-1}$.
\end{lem}

\begin{proof}
The computations are straightforward.
We note that by Girsanov's formula (\ref{eq: girsanov}) and the definitions of $\mu^j$ and $\Sigma^j$ we have
\begin{equation}\label{eq: gg}
p(x^T\given j, \theta^j, s^2) = e^{(\theta^j)^T\mu^j -\frac12(\theta^j)^T\Sigma^j \theta^j}
\end{equation}
and by construction of the prior,
\[
p(\theta^j \given j, s^2) \propto (s^2)^{-\frac{m_j}2}e^{-\frac12 (\theta^j)^T(s^2\Xi^j)^{-1} \theta^j},
\quad p(s^2) \propto (s^2)^{-a-1}e^{-\frac b{s^2}}.
\]
It follows that
\[
p(\theta^j  \given  s^2, j, x^T)
\propto p(x^T\given j, \theta^j, s^2) p(\theta^j \given j, s^2)
\propto e^{(\theta^j)^T\mu^j -\frac12(\theta^j)^TW^j \theta^j},
\]
which proves the first assertion of the lemma.
Next we write
\begin{align*}
p(s^2  \given \theta^j, j, x^T) &\propto p(x^T\given j, \theta^j, s^2) p(\theta^j \given j, s^2)p(s^2)
\\&\propto (s^2)^{-m_j/2 - a - 1} \exp\left(-\frac{-b - (1/2) (\th^{j})^T (\Xi^j)^{-1} \th^j}{s^2}\right),
\end{align*}
which yields the second assertion.
\end{proof}

The lemma shows that Gibbs sampling can be used to sample (approximately)
from the continuous observations
posteriors of $s^2$ and $\theta^j$ within a fixed model $j$. In the next subsections we explain how to combine
this with reversible jump steps between different models and with data augmentation through the generation of
diffusion bridges in the case of discrete-time observations.

\subsection{Reversible jumps between models}

In this subsection we still assume that  we have continuous data $x^T = (x_t: t \in [0,T])$
at our disposal. We complement the within-model computations given by
Lemma \ref{lem:postcoeff} with (a basic version of) reversible jump MCMC (cf.\ \cite{Green}) to explore different models.
We will construct a Markov chain which has the full posterior $p(j, \theta^j, s^2 \given x^T)$
as invariant distribution and hence can be used to generate approximate draws from the posterior distribution of
the drift function $b$.

We use an auxiliary Markov chain on $\NN$ with transition probabilities $q(j'\given j)$, $j, j' \in \NN$.
As the notation suggests, we denote by $p(j \given s^2,  x^T)$ the conditional (posterior) probability
of model $j$ given the parameter $s^2$ and  the data $x^T$.
Recall that $p(j)$ is the prior probability of model $j$.
Now we define the quantities
\begin{align}\label{eq:rel}
 B(j' \given j) &= \frac{p(x^T\given j', s^2)}{p(x^T\given j, s^2)},\\
  R(j'\given j)  & =  \frac{p(j')q(j\mid j')}{p(j)q(j' \mid j)}.
 \end{align}
Note that  $B(j'\given j)$ is the  Bayes factor   of model $j'$ relative to model $j$,
for a fixed scale $s^2$.
To simplify the notation, the dependence of this quantity on $s^2$ and $x^T$ is suppressed.

The overall structure of the algorithm that we propose is that of a
componentwise Metropolis--Hastings (MH) sampler. The scale parameter $\scale$ is taken as component I and
the pair $(j, \th^j)$ as component II.  Starting with some initial value $(j_0, \th^{j_0}, s^2_0)$,
alternately moves from $(j, \th^j, \scale)$ to $(j, \th^{j}, (s')^2)$ and moves from
$(j, \th^j, \scale)$ to $(j', \th^{j'}, \scale)$ are performed, where in each case the other component remains unchanged.

Updating the first component is done with a simple Gibbs move, that is a new value for $\scale{}$ is sampled from its posterior
distribution described by
Lemma \ref{lem:postcoeff},  given the current value of the remaining parameters.

\begin{move}[I] Update the scale. Current state:  $(j,\th^j, s^2)$.
\begin{itemize}
\item Sample $(s')^2 \sim {\rm IG}(a + (1/2) m_j, b + (1/2) (\th^{j})^T (\Xi^j)^{-1} \th^j )$.
\item Update the state to  $(j,\th^j, (s')^2)$.
\end{itemize}
\end{move}

 The second component $(j, \theta^j)$ has varying dimension and a reversible jump move is employed to ensure detailed balance
 (e.g.\ \cite{Green2003}, \cite{ChapmanHall}).  To perform a transdimensional move, first a new model $j'$ is chosen and a sample from the posterior for $\theta^j$ given by Lemma \ref{lem:postcoeff} is drawn.

\begin{move}[II] Transdimensional move. Current state: $(j,\th^j, \scale)$.
\begin{itemize}
\item Select a new model $j'$ with probability $q(j' \mid j)$.
\item Sample $\th^{j'} \sim N_{m_j'}((W^{j'})^{-1}\mu^{j'}, (W^{j'})^{-1})$.
\item Compute $r = \bquot\rquot.$
\item With probability $\min\{1, r\}$ update the state to
$(j', \th^{j'}, \scale{})$, else leave the state  unchanged.
\end{itemize}
\end{move}

All together we have now constructed a Markov chain $Z_0, Z_1, Z_2, \ldots$ on the transdimensional space
$E = \bigcup_{j \in \NN}\, \{j\}\times \RR^{m_j}\times  (0,\infty)$, whose evolution can be described as follows.

\bigskip

\renewcommand{\arraystretch}{1.2}

\begin{center}
\begin{tabular}{|p{0.7\textwidth}|}
\hline
{\bf Continuous observations algorithm}\\
\hline
\hline
{\em Initialization}. Set $Z_0  =  (j_0, \theta^{j_0}, s^2_0)$.\\
{\em Transition}. \quad Given the current state $Z_j = (j, \theta^{j}, s^2)$:\\
$\bullet$ Sample $(s')^2 \sim {\rm IG}(a + (1/2) m_j, b + (1/2) (\th^{j})^T (\Xi^j)^{-1} \th^j )$.\\
$\bullet$ Sample $j' \sim q(j' \given j)$,\\
$\bullet$ Sample $\th^{j'} \sim N_{m_j'}((W^{j'})^{-1}\mu^{j'}, (W^{j'})^{-1})$. \\
$\bullet$ With probability $\min\{1,\bquot\rquot\}$,  set \\ \quad
$Z_{j+1} = (j', \theta^{j'}, (s')^2)$, else set
$Z_{j+1} = (j, \theta^{j}, (s')^2)$.\\
\hline
\end{tabular}
\end{center}

\bigskip

\noindent
Note that $r = ({q(j \mid j')}/{q(j \mid j')})({p(j' \mid s^2, x^T)}/{p(j' \mid s^2, x^T)})$,  so effectively we perform Metropolis--Hastings for updating $j$.  As a consequence, the vector of coefficients $\th^{j'}$ needs to be drawn only in case the proposed  $j'$ is accepted.

The following lemma asserts that the constructed chain indeed has the desired stationary
distribution.

\begin{lem} \label{lem: 1}
The Markov chain $Z_0, Z_1, \ldots$ has the posterior
$p(j, \theta^j, s^2 \given x^T)$ as invariant  distribution.
\end{lem}

\begin{proof}
By Lemma \ref{lem:postcoeff} we have that in move I the chain moves from state $(j, \theta^j, s^2)$ to
$(j, \theta^{j}, s'^2)$ with probability
$p(s'^2 \given j, \theta^j, x^T)$.
Conditioning shows that we have detailed balance for this move, that is,
\begin{equation}\label{eq: balance}
p(j, \theta^j, s^2\given x^T) p(s'^2 \given j, \theta^j, x^T) =
p(j, \theta^j, s'^2\given x^T) p(s^2 \given j, \theta^j, x^T).
\end{equation}

In view of Lemma \ref{lem:postcoeff} again, the probability that the  chain moves from state $(j, \theta^j, s^2)$ to
$(j', \theta^{j'}, s^2)$ in move II equals, by construction,
\[
p((j, \theta^j, s^2)  \to (j', \theta^{j'}, s^2)) =
\min\Big\{1, \frac{q(j\mid j')}{q(j' \mid j)}\frac{p(j' \given s^2, x^T)}{p(j \given s^2, x^T)}\Big\}
 q(j'\given j)p(\theta^{j'} \given j', s^2, x^T).
\]
Now suppose first that the minimum is less than $1$. Then using
\[
p(j, \theta^j, s^2\given x^T) = p(\theta^j\given j, s^2, x^T) p(j \given s^2, x^T)p(s^2\given x^T)
\]
and
\[
p(j', \theta^{j'}, s^2\given x^T) = p(\theta^{j'}\given j', s^2, x^T) p(j' \given s^2, x^T)p(s^2\given x^T)
\]
it is easily verified that we have the detailed balance relation
\begin{align*}
 p(j, \theta^j, s^2\given x^T) p((j, \theta^j, s^2)  \to (j', \theta^{j'}, s^2))
 = p(j', \theta^{j'}, s^2\given x^T) p((j', \theta^{j'}, s^2)  \to (j, \theta^{j}, s^2)).
\end{align*}
for move II. The case that the minimum is greater than $1$ can be dealt with similarly.

We conclude that   we have detailed balance for both components of our  MH sampler. Since our algorithm is a variable-at-a-time Metropolis--Hastings algorithm, this
implies the statement of the lemma (see for example section 1.12.5 of \cite{ChapmanHall}).
\end{proof}

\subsection{Data augmentation for discrete data}
\label{dataaug}

So far we have been dealing with continuously observed diffusion. Obviously, the phrase ``continuous data'' should be interpreted properly.
In practice it means that the frequency at which the diffusion is observed is so high that
the error that is incurred by approximating the  quantities (\ref{eq:mu}) and (\ref{eq:sigma})
by their empirical counterparts, is negligible.
If we only have low-frequency, discrete-time observations at our disposal, these approximation errors can
typically not be ignored however and can  introduce undesired biases.
In this subsection we  explain how our algorithm can be extended to accommodate this situation as well.

We assume now that we only have partial observations
$x_0, x_\Delta, \dots, x_{n\Delta}$ of our diffusion process, for some $\Delta > 0$ and $n \in \NN$.
We set $T = n\Delta$.
The discrete observations constitute a Markov chain, but it is well known that
the transition densities of discretely observed diffusions
and hence the likelihood are not available in closed form in general. This complicates a Bayesian analysis.
An approach that has been proven to be
very fruitful,  in particular  in the context  of parametric estimation for discretely observed diffusions, is to view the continuous diffusion segments
between the observations as missing data and to treat them as latent (function-valued) variables.
Since the continuous-data likelihood is known (cf.\ the preceding subsection),
data augmentation methods (see \cite{TannerWong}) can be used to circumvent the unavailability of the likelihood
in this manner.

As discussed in \cite{VdMeulenVZanten} and shown in a practical setting by  \cite{Pokern}, the data augmentation approach is not limited to the parametric
setting  and can be used in the present nonparametric problem as well.
Practically it involves appending an extra step to the algorithm presented in the preceding subsection,
corresponding to  the simulation of the appropriate diffusion bridges.
If we denote again the continuous observations by $x^T = (x_t: t \in [0,T])$ and the
discrete-time observations by $x_\Delta, \ldots, x_{n\Delta}$, then using the same notation as above we essentially
want to sample from the conditional distribution
\begin{equation}\label{eq: piet}
p(x^T \given j, \theta^j, s^2, x_\Delta, \ldots, x_{n\Delta}).
\end{equation}
Exact simulation methods have been proposed in the literature to accomplish this, e.g.\ \cite{BeskosPapaspiliopoulosRoberts},
\cite{BeskosPapaspiliopoulosRobertsFearnhead}.
For our purposes exact simulation is not strictly necessary  however and it is
more convenient to add a Metropolis--Hastings step corresponding to a Markov chain that has the diffusion bridge
law given by (\ref{eq: piet}) as stationary distribution.

 Underlying the MH sampler for diffusion bridges  is the fact that by Girsanov's theorem,
the conditional distribution of the continuous segment $X^{(k)} = (X_t: t \in [(k-1)\Delta, k\Delta])$
given that $X_{(k-1)\Delta} = x_{(k-1)\Delta}$ and $X_{k\Delta} = x_{k\Delta}$, is equivalent to the distribution of a Brownian bridge
that goes from $x_{(k-1)\Delta}$ at time $(k-1)\Delta$ to $x_{k\Delta}$ at time $k\Delta$.
The corresponding Radon-Nikodym derivative is proportional to
\begin{equation}
\label{eq:condlike}
L_k(X^{(k)}\given b) =
	\exp\Big( \int_{(k-1)\Delta}^{k\Delta} b(X_t) \dd X_t- \frac12 \int_{(k-1)\Delta}^{k\Delta} b^2(X_t) \dd t\Big).
\end{equation}
We also note that due to the Markov property of the diffusion, the different diffusion bridges $X^{(1)}, \ldots,
X^{(n)}$, can be dealt with independently.

Concretely, the missing segments $x^{(k)} = (x_t: t \in ((k-1)\Delta, k\Delta))$, $k =1, \ldots, n$
can be added
as latent variables to the Markov chain constructed in the preceding subsection, and the following
move has to be added to moves I and II introduced above. It is a standard Metropolis--Hastings
step for the conditional law (\ref{eq: piet}),
with independent Brownian bridge proposals. For more details on this type of MH samplers for diffusion bridges
we refer to  \cite{RobertsStramer}.

\begin{move}[III]  Updating the diffusion bridges. Current state:
 $(j,\th^{j},\scale, x^{(1)}, \ldots, x^{(n)})$:
\begin{itemize}
  \item For $k =1, \ldots, n$, sample a Brownian bridge $w^{(k)}$ from
  $((k-1)\Delta, x_{(k-1)\Delta})$ to $(k\Delta, x_{k\Delta})$.
  \item For  $k  = {1, \dots, n}$, compute  $r_k = L_k(w^{(k)}\given b) /L_k(x^{(k)} \given b)$, for
  $b = \sum_{l \le m_j} \theta^j_l\psi_l$.
  \item Independently, for $k =1, \ldots, n$, with probability $\min \{1, r_k\}$ update $x^{(k)}$ to $w^{(k)}$, else retain $x^{(k)}$.
\end{itemize}
\end{move}

Of course the segments $x^{(1)}, \ldots, x^{(n)}$ can always be concatenated to yield a continuous function on $[0,T]$.
In this sense move III can be viewed as a step that generates new, artificial continuous data given the discrete-time data.
It is convenient to consider this whole continuous path on $[0,T]$ as the latent variable. When the new move is
combined with the ones defined earlier
a Markov chain  $\tilde Z_0, \tilde  Z_1, \tilde  Z_2, \ldots$ is obtained  on the space
$\tilde E = \bigcup_{j \in \NN}\, \{j\}\times \RR^{m_j}\times  (0,\infty) \times C[0,T]$.
Its evolution can be described as follows.

\bigskip

\begin{center}
\begin{tabular}{|p{0.9\textwidth}|}
\hline
{\bf Discrete observations algorithm}\\
\hline
\hline
{\em Initialization}. Set $\tilde  Z_0  =  (j_0, \theta^{j_0}, s^2_0, x_0^T)$,
where $x_0^T$ is for instance\\ obtained by linearly interpolating the observed data points.\\
{\em Transition}. \quad Given the current state $\tilde Z_j = (j, \theta^{j}, s^2, x^T)$, construct $\tilde  Z_{j+1}$
as follows:\\
$\bullet$ Sample $(s')^2 \sim {\rm IG}(a + (1/2) m_j, b + (1/2) (\th^{j})^T (\Xi^j)^{-1} \th^j )$, update $s^2$ to $(s')^2$.\\
$\bullet$ Sample $j' \sim q(j' \given j)$ and $\th^{j'} \sim N_{m_j'}((W^{j'})^{-1}\mu^{j'}, (W^{j'})^{-1})$. \\
$\bullet$ With probability $\min\{1,\bquot\rquot\}$,  update $(j, \theta^j)$ to $(j', \theta^{j'})$,
 else retain  $(j, \theta^{j})$.\\
 $\bullet$ For $k =1, \ldots, n$, sample a Brownian bridge $w^{(k)}$ from
  $((k-1)\Delta, x_{(k-1)\Delta})$ \\ \quad to $(k\Delta, x_{k\Delta})$ and
compute  $r_k = L_k(w^{(k)}\given b) /L_k(x^{(k)} \given b)$, for
  $b = \sum_{l \le m_j} \theta^j_l\psi_l$.\\
  $\bullet$ Independently, with probability $\min \{1, r_k\}$, update $x^{(k)}$ to $w^{(k)}$, else retain $x^{(k)}$.\\
\hline
\end{tabular}
\end{center}

\bigskip

\noindent
It follows from the fact that move III is a MH step for the conditional law (\ref{eq: piet}) and
Lemma \ref{lem: 1} that the new chain has the correct  stationary
distribution again.

\section{Simulation results}
\label{sec: sim}

The implementation of the algorithms presented in the preceding section involves
the computation of several quantities, including the Bayes factors $B(j'\given j)$
and sampling from the posterior distribution of $\theta^j$ given $j$ and $s^2$.
In  Section \ref{sec: num} we explain in some detail how these
issues can be tackled efficiently.  In the present
section we first investigate the performance of our method on simulated data.

For the drift function, we first choose the function $b(x)=12(a(x)+0.05)$ where
\begin{equation}
a(x)=\begin{cases} \frac27-x-\frac27(1-3x) \sqrt{|1-3x|} & x \in [0,2/3) \\
				-\frac27+\frac2{7}x & x \in [2/3,1]
\end{cases}
\end{equation}
This function is H\"older-continuous of order $1.5$ on $[0,1]$. A plot of $b$ and its derivative is shown in Figure \ref{fig:drift}.
Clearly, the derivative is not differentiable in $0$, $1/3$ and $2/3$.

\begin{figure}
\begin{center}
\includegraphics[width=.75\linewidth]{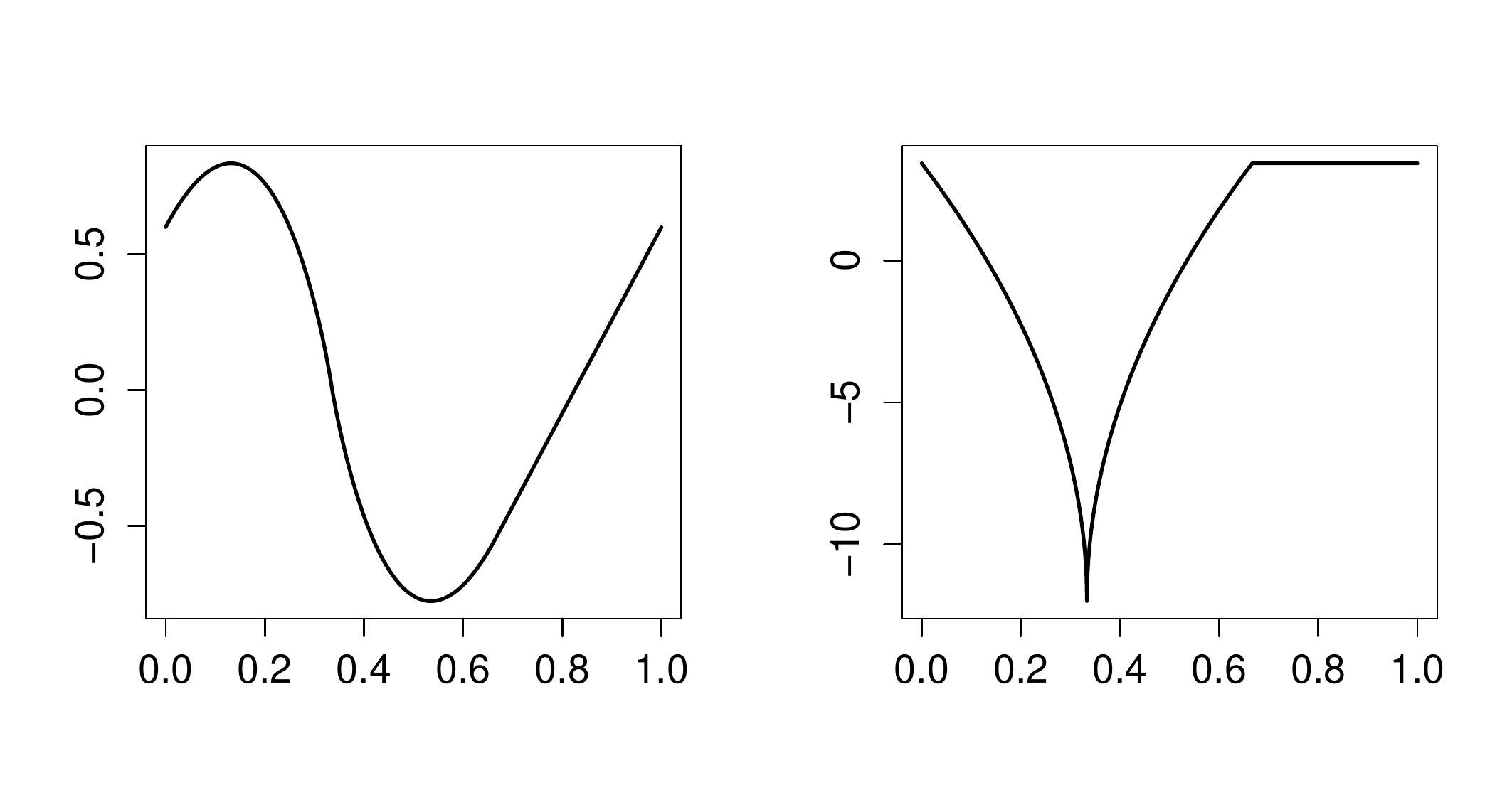}
\caption{Left: drift function. Right: derivative of drift function}\label{fig:drift}
\end{center}
\end{figure}
\begin{figure}[h]
\centering
\includegraphics[width=.8\linewidth]{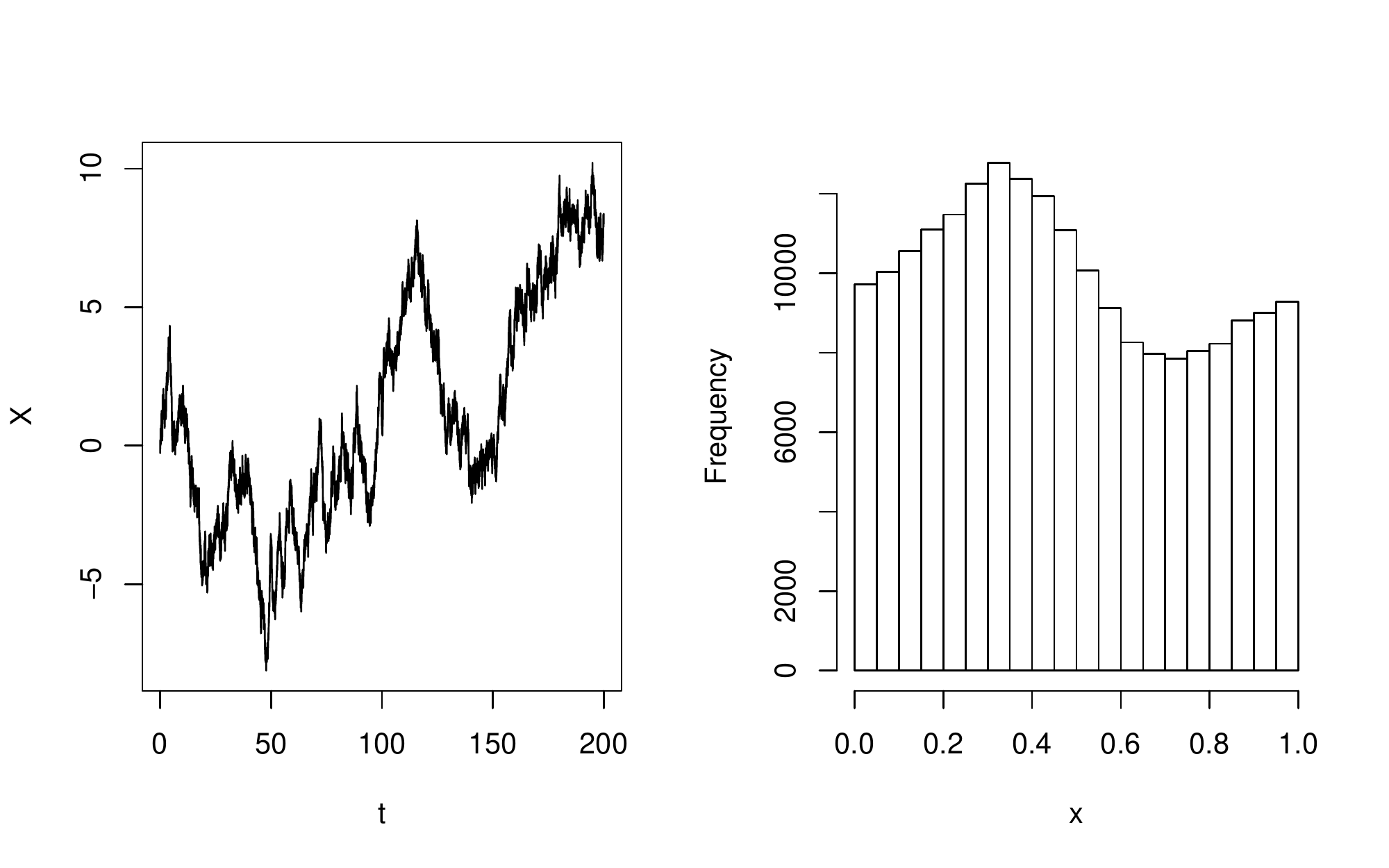}
\caption{Left: simulated data. Right: histogram of simulated data modulo $1$. }
\label{fig:data}
\end{figure}

\begin{figure}
\begin{center}
\includegraphics[width=.4\linewidth]{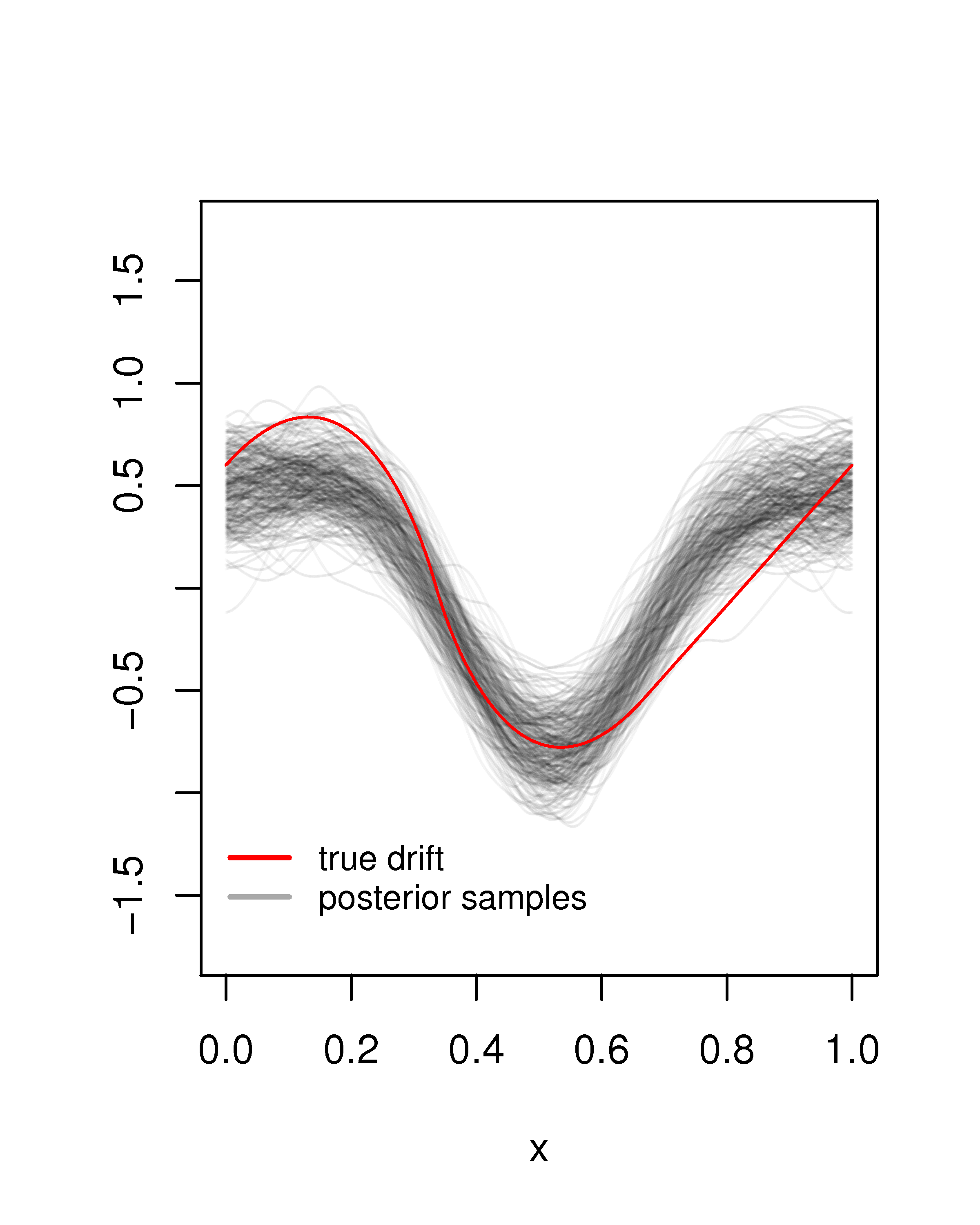}
\includegraphics[width=.4\linewidth]{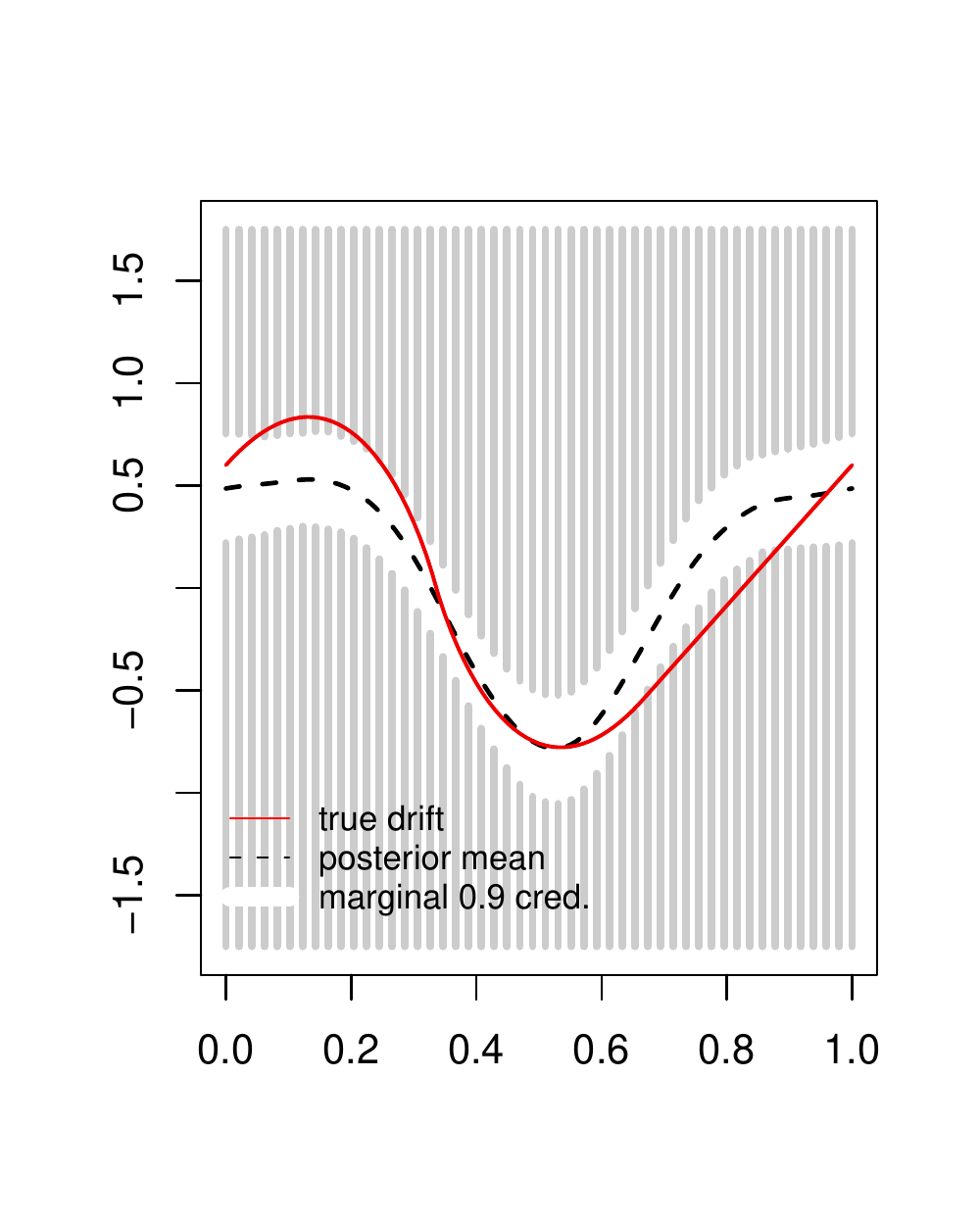}
\caption{Left: true drift function (red, solid) and samples from the posterior distribution. Right: drift function (red, solid), posterior mean (black, dashed) and $90 \%$ pointwise credible bands}\label{fig:1}
\end{center}
\end{figure}

We simulated a diffusion  on the time interval $[0,200]$ using the Euler discretization scheme with time discretization step equal to  $10^{-5}$. Next, we retained all observations at values $t=i\Delta$ with $\Delta=0.001$ and $i=0,\ldots, 200.000$.  The data are shown in Figure \ref{fig:data}. From the histogram we can see that the process spends most time near $x=1/3$, so we expect  estimates for the drift to be best in this region.  For now, we consider the data as essentially continuous time data, so no  data-augmentation scheme is employed.

We define a prior  with Fourier basis functions as described in Section \ref{sec: fourier}, choosing  regularity $\beta = 1.5$. With this choice, the regularity of the prior matches that of the true drift function.

For the reversible jump algorithm there are a few  tuning parameters.  For the model generating chain, we choose
$	q(j\mid j)=1/2$, $ q(j+1 \mid j)=q(j-1\mid j)=1/4$.
For the priors on the models we chose $C=-\log(0.95)$ which means that
$p(j)\propto (0.95)^{m_j}$ expressing weak prior belief in a course (low) level model. For the inverse Gamma prior on the scale we take the hyper parameters
$a = b = 5/2$.

We ran the continuous time algorithm for $3000$ cycles and discarded the first $500$ iterations as burn-in. We fix this number of iterations and burn-in  for all other MCMC simulations.  In Figure \ref{fig:1} we  show the resulting
posterior mean (dashed curve) and $90\%$ pointwise credible intervals (visualized by the gray bars). The posterior mean was estimated using Rao-Blackwellization (\cite{Robert}, section 4.2).
(Specifically, the posterior mean drift was not computed as the pointwise average of the drift functions $b$ sampled at each iteration.
Rather, the average of the posterior means $(W^{j'})^{-1}\mu^{j'}$ obtained at each MCMC iteration (see move II) was used.)

Insight in the mixing properties of the  Markov chain is gained by considering  traces of the sampled drift function at several fixed points as shown in Figure \ref{fig:Atrace}.
 The trace plots indicate that indeed the first 200-300 iterations should be discarded.  Plots of the visited models over time and the corresponding acceptance probabilities are shown in figures \ref{fig:Amodels} and \ref{fig:accep} respectively. The  mean and median of the scaling parameter $s^2$ are given by $1.91$ and $1.64$ respectively (computing upon first removing burn-in samples).
\begin{figure}
\centering
\includegraphics[width=.9\linewidth]{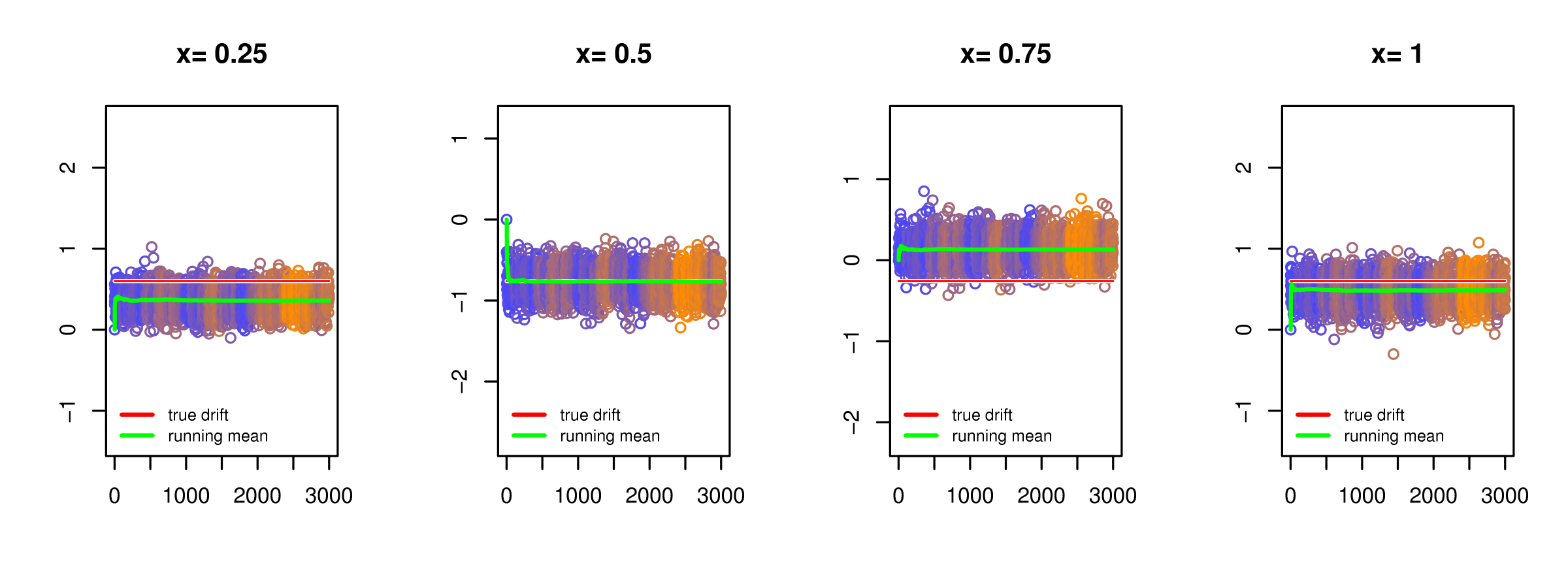}
\caption{Trace and running mean of the sampled drift at different design points. The color of the samples indicates the current model, cold colors correspond to small values of $j$.}
\label{fig:Atrace}
\end{figure}

\begin{figure}
\centering
\includegraphics[width=.7\linewidth]{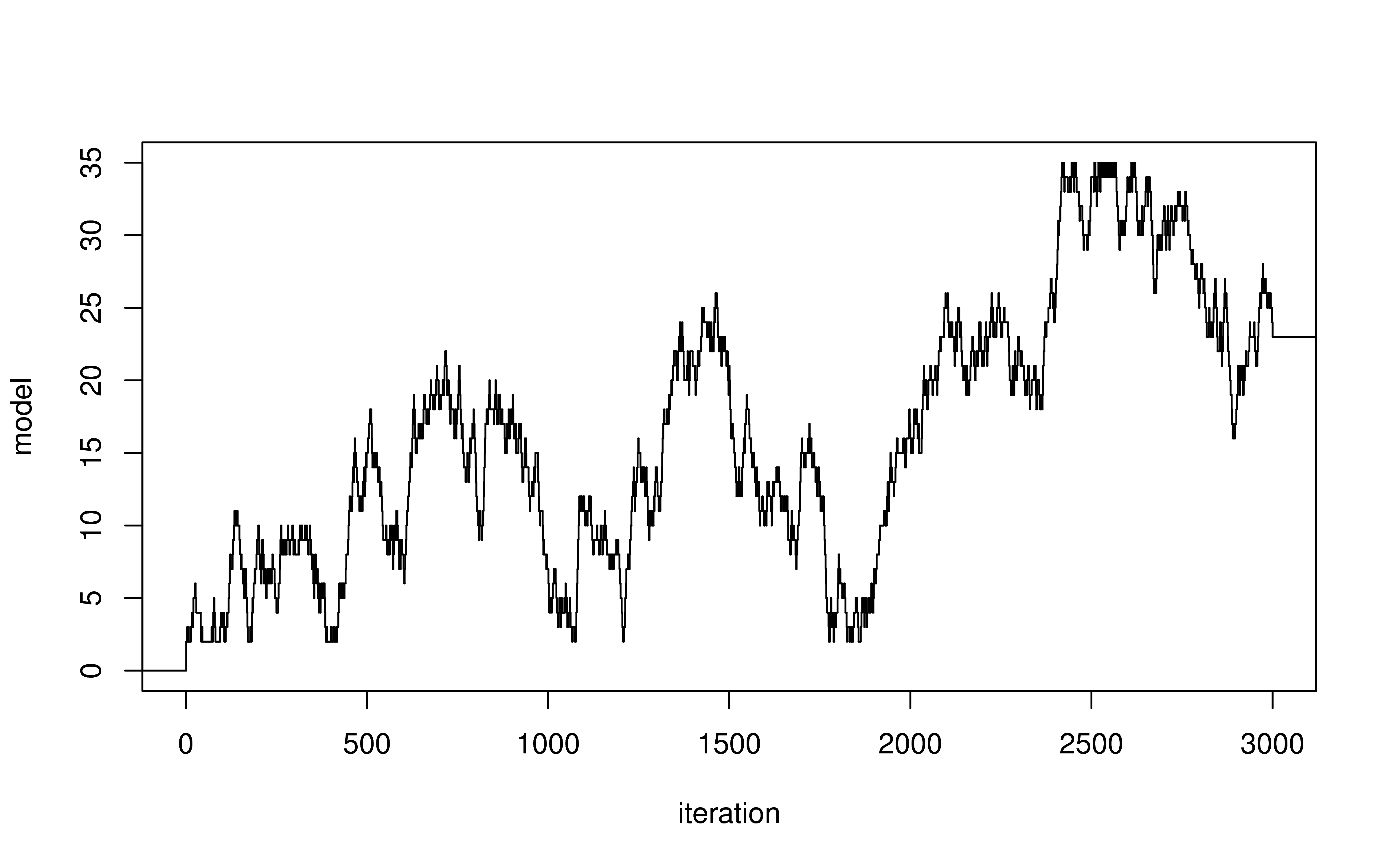}
\caption{Models visited over time.}
\label{fig:Amodels}
\end{figure}

\begin{figure}
\centering
\includegraphics[scale=0.6]{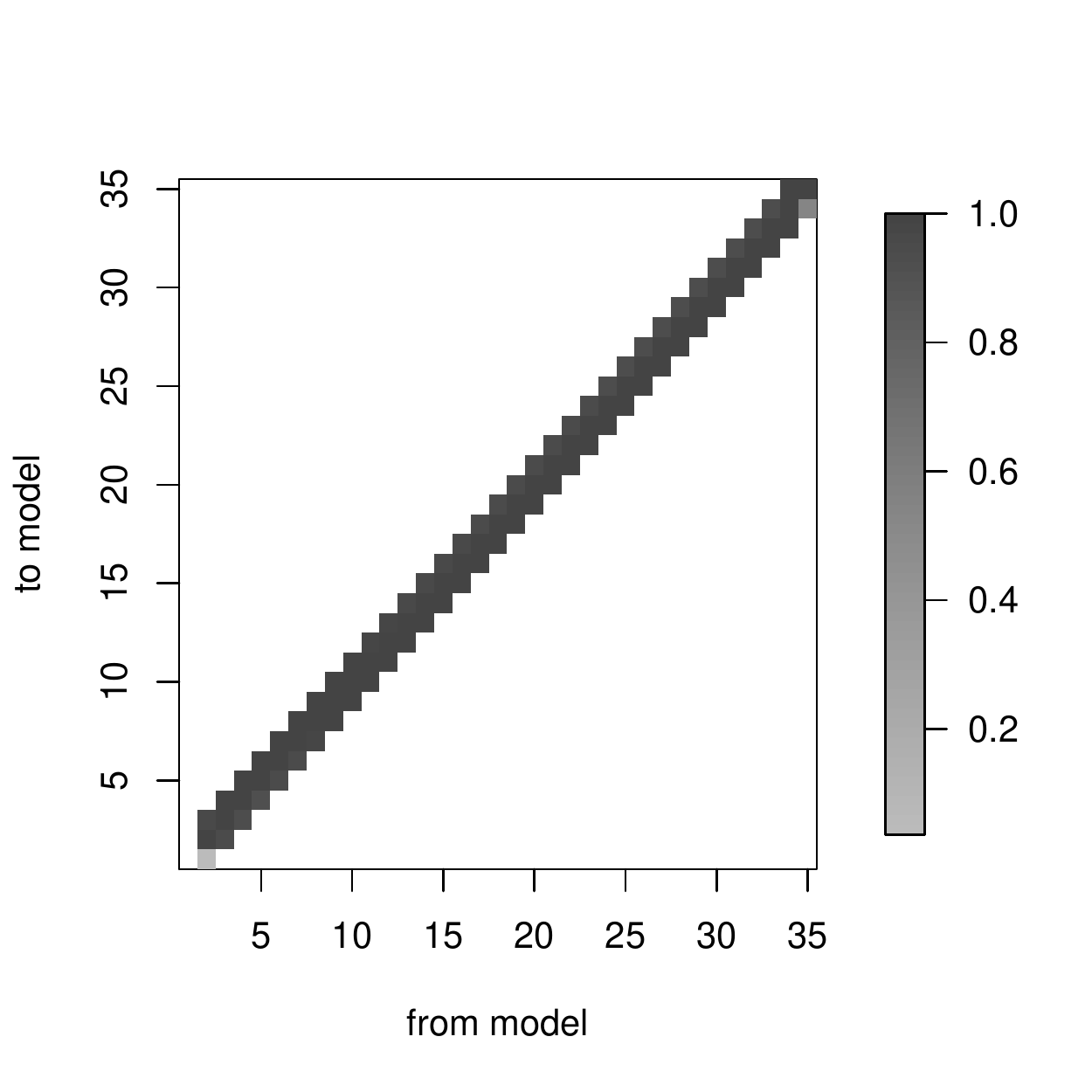}
\caption{Average acceptance probabilities for moves between models.}
\label{fig:accep}
\end{figure}

To judge the algorithm with respect to the sensitivity of $C$, we ran the same algorithm as well for  $C=0$. The latter choice is often made and  reflects equals prior belief  on all models under consideration. If $C=0$, then the chain spends more time in higher models. However, the posterior mean and pointwise credible bands are practically equal to the case $C=-\log(0.95)$.

To get a feeling for the sensitivity of the results on the choice of the hyper-parameters $a$ and $b$,  the analysis was redone for $a=b=5$. The resulting posterior mean and credible bands turned out to be indistinguishable from the case $a=b=2.5$.

Clearly, if we would have chosen an example with less data, then the influence of the priors would be  more strongly pronounced in the results.

\subsection{Effect of the prior on the model index}
If, as in the example thus far, the parameter  $\beta$ is chosen to match the regularity of the true drift,
one would expect that using a prior where the truncation point for in the series expansion of  the drift
is fixed at a high enough level,
one would  get  results comparable to those we obtained in Figure \ref{fig:1}.
 If we fix the level at $j=30$, this indeed  turns out to be the case.  The main advantage of putting a prior
 on the truncation level and using a reversible jump algorithm is an improvement in computing time.  For this example, a simulation run  for the reversible jump model took about $55\%$ of that for the  fixed dimensional model. For the reversible jump run,  the average of the (non-burn-in) samples of the model index $j$ equals $18$.

\subsection{Effect of the random scaling parameter}\label{sec:scale}
To assess the effect of including a prior on the multiplicative  scaling parameter, we ran the same simulation as before, though keeping $s^2$ fixed to either $0.25$ (too small) or $50$ (too large).
  For both cases, the posterior mean, along with $90\%$ pointwise credible bounds, is depicted in Figure \ref{fig:scalevarying}.
\begin{figure}
\begin{center}
\includegraphics[width=.3\linewidth]{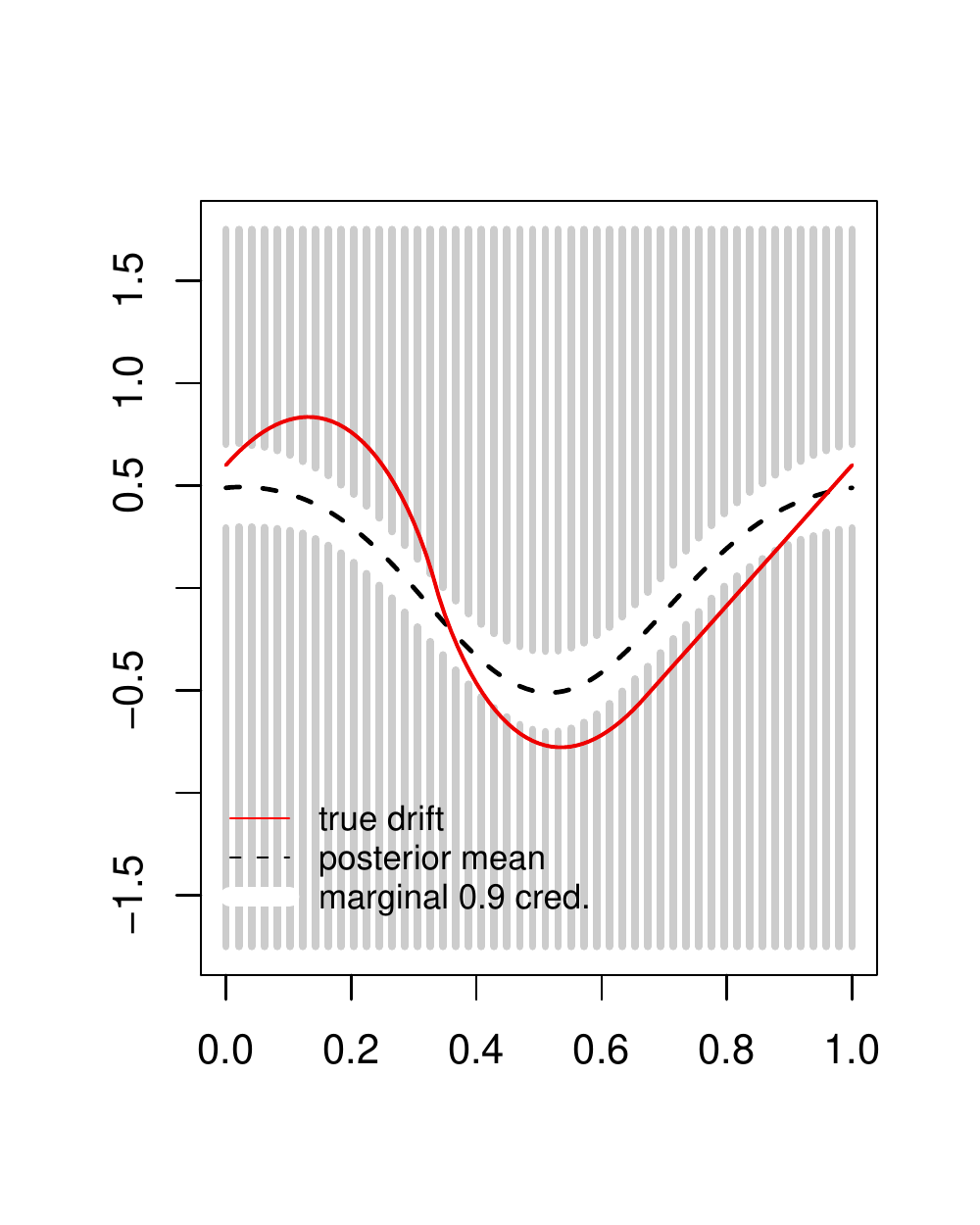}
\includegraphics[width=.3\linewidth]{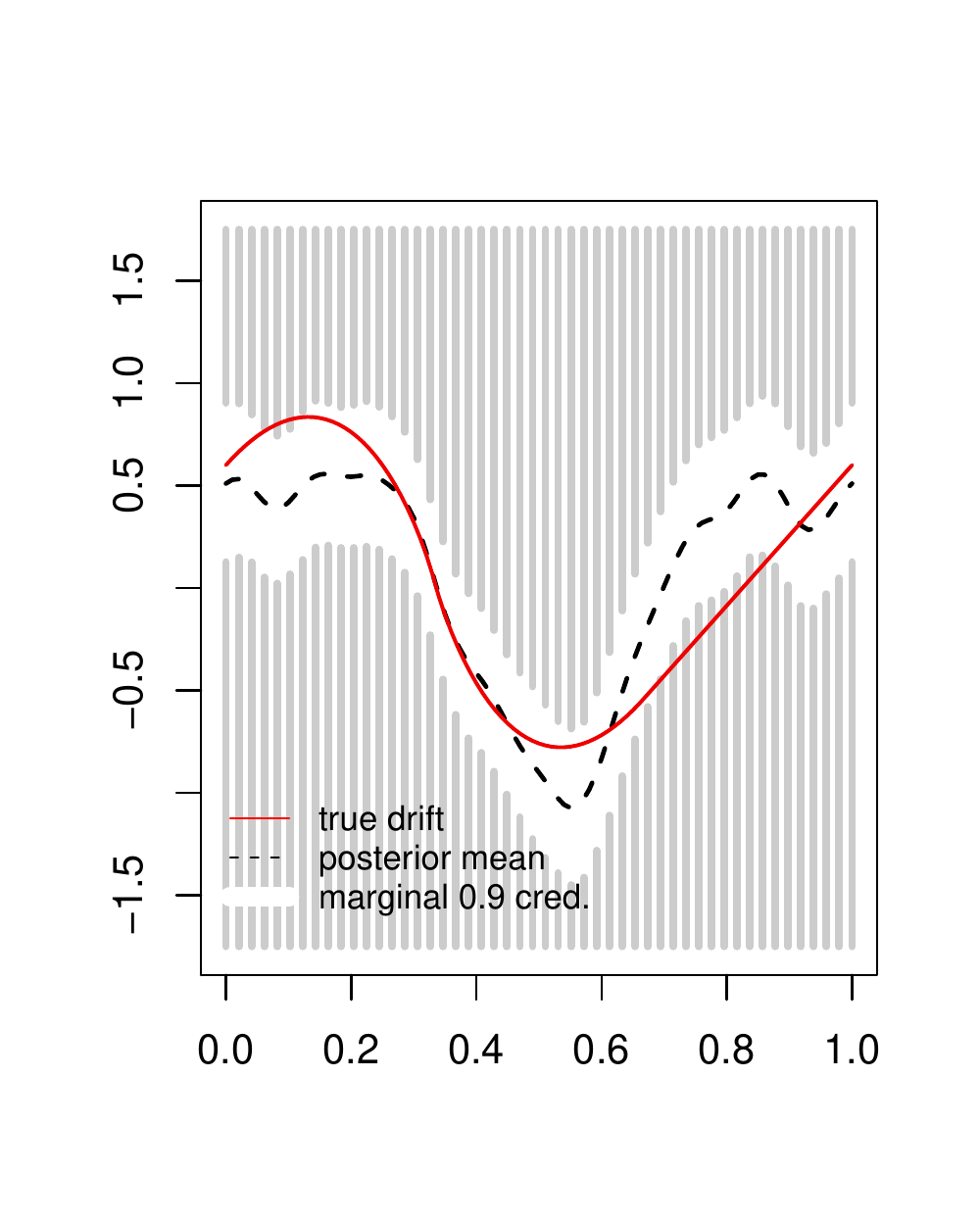}
\includegraphics[width=.3\linewidth]{images/bb-fou-1,5-periodic-cont3000final122e+05in200post.pdf}
\caption{Drift function (red, solid), posterior mean (black, dashed) and $90 \%$ pointwise credible bands. Left: $s^2=0.25$ fixed. Middle: $s^2=50$ fixed. Right: random scaling parameter (Inverse gamma prior with $a=b=2.5$). }
\label{fig:scalevarying}
\end{center}
\end{figure}
For ease of comparison, we added the right-hand-figure of Figure \ref{fig:1}.  Clearly, fixing $s^2=0.25$ results in oversmoothing.  For $s^2=50$, the credible bands are somewhat wider and suggest more fluctuations in the drift function than are actually present.

\subsection{Effect of misspecifying smoothness of the prior}
If we consider the prior without truncation, then the smoothness of the prior is essentially governed by the decay of the variances on the coefficients. This decay is determined by the value of $\beta$. Here, we investigate the effect of misspecifying $\beta$. We consider $\beta \in \{0.25, 1.5, 3\}$. In Figure  \ref{fig:effbeta_rj}  one can assess the difference in posterior mean, scaling parameter and  models visited for these three values of $\beta$.  Naturally, if $\beta$ is too large, higher level models and relatively large values for $s^2$ are chosen to make up for the overly fast decay on the variances of the coefficients. Note that the boxplots are for $\log(s^2)$, not $s^2$.

It is interesting to investigate  what happens if the   same analysis is done  without the reversible jump algorithm, thus fixing a high level truncation point. The results are in Figure \ref{fig:effbeta_fixlevel}. From this figure, it is apparent that if  $\beta$ is too small the posterior mean is very wiggly. At the other extreme, if $\beta$ is too large, we are oversmoothing and the true drift is outside the credible bands except near $x=1/3$.  As such, misspecifying $\beta$ can result in very bad estimates if a high level is fixed.
From the boxplots in Figures \ref{fig:effbeta_rj} and \ref{fig:effbeta_fixlevel} one can see that the larger $\beta$, the larger the scaling parameter. Intuitively, this makes sense. Moreover, in case we employ a reversible jump algorithm, a too small (large) value for $\beta$ is compensated for by low (high) level models.

\begin{figure}
\begin{center}
\includegraphics[width=.3\linewidth]{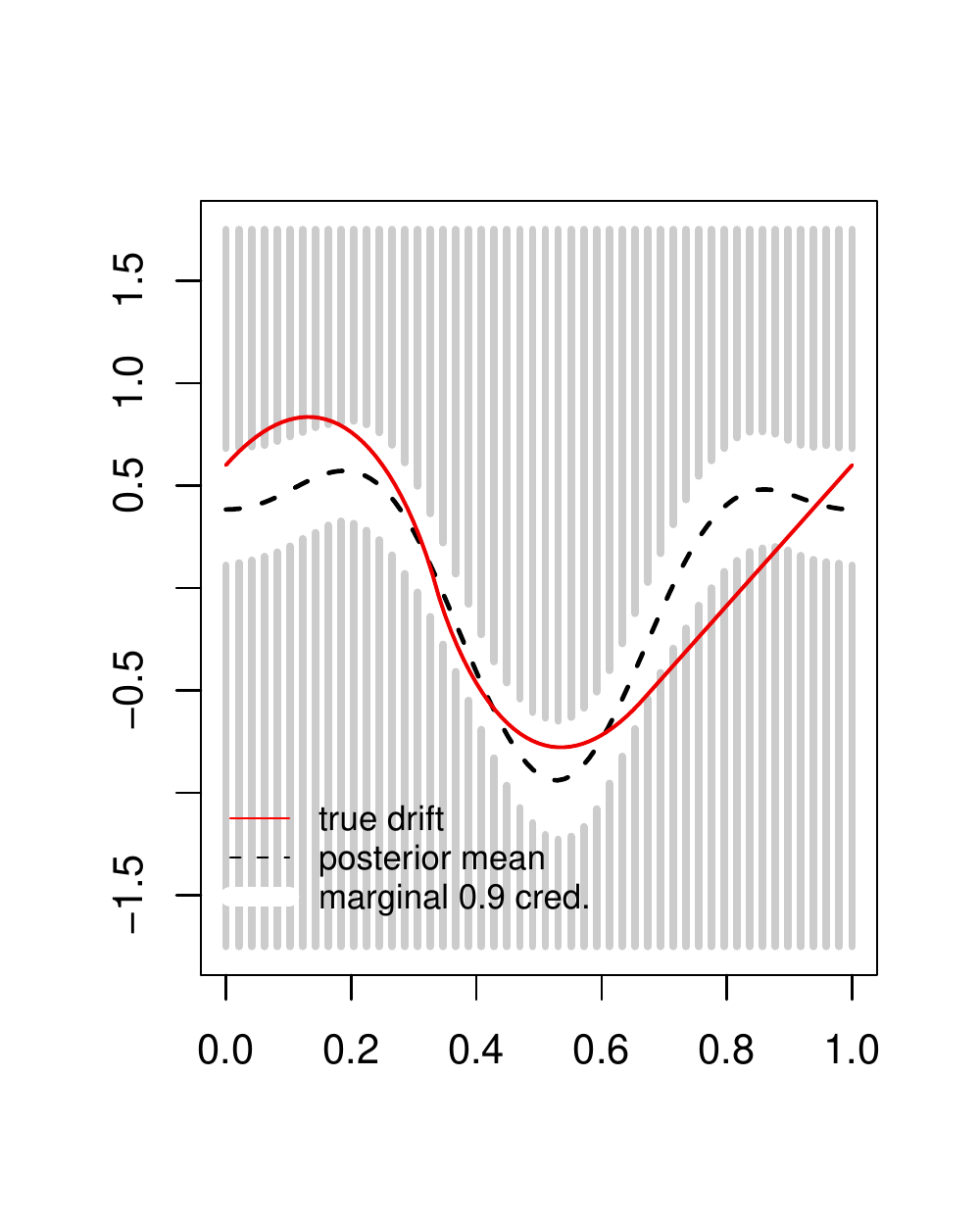}
\includegraphics[width=.3\linewidth]{images/bb-fou-1,5-periodic-cont3000final122e+05in200post}
\includegraphics[width=.3\linewidth]{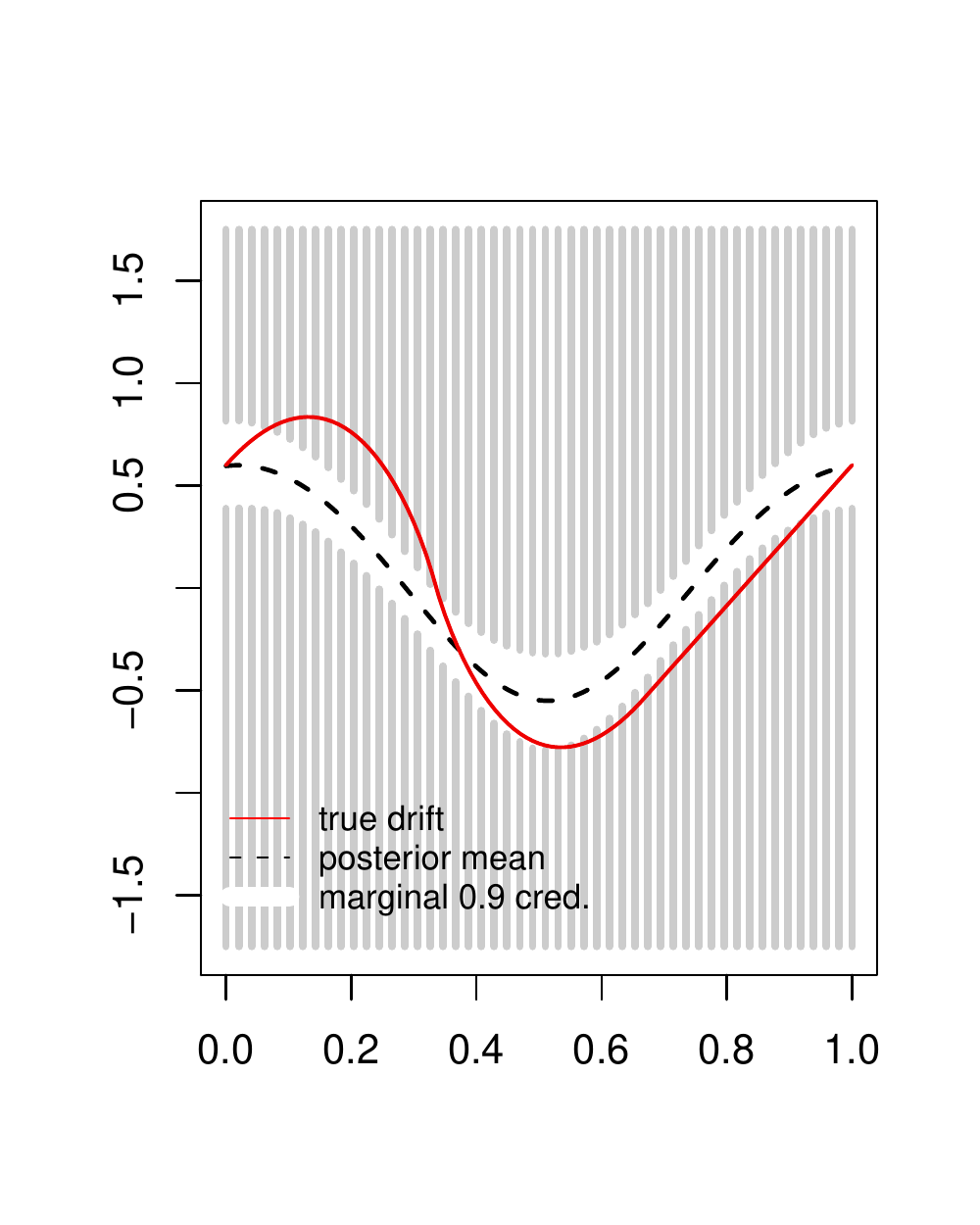}

\includegraphics[width=.8\linewidth]{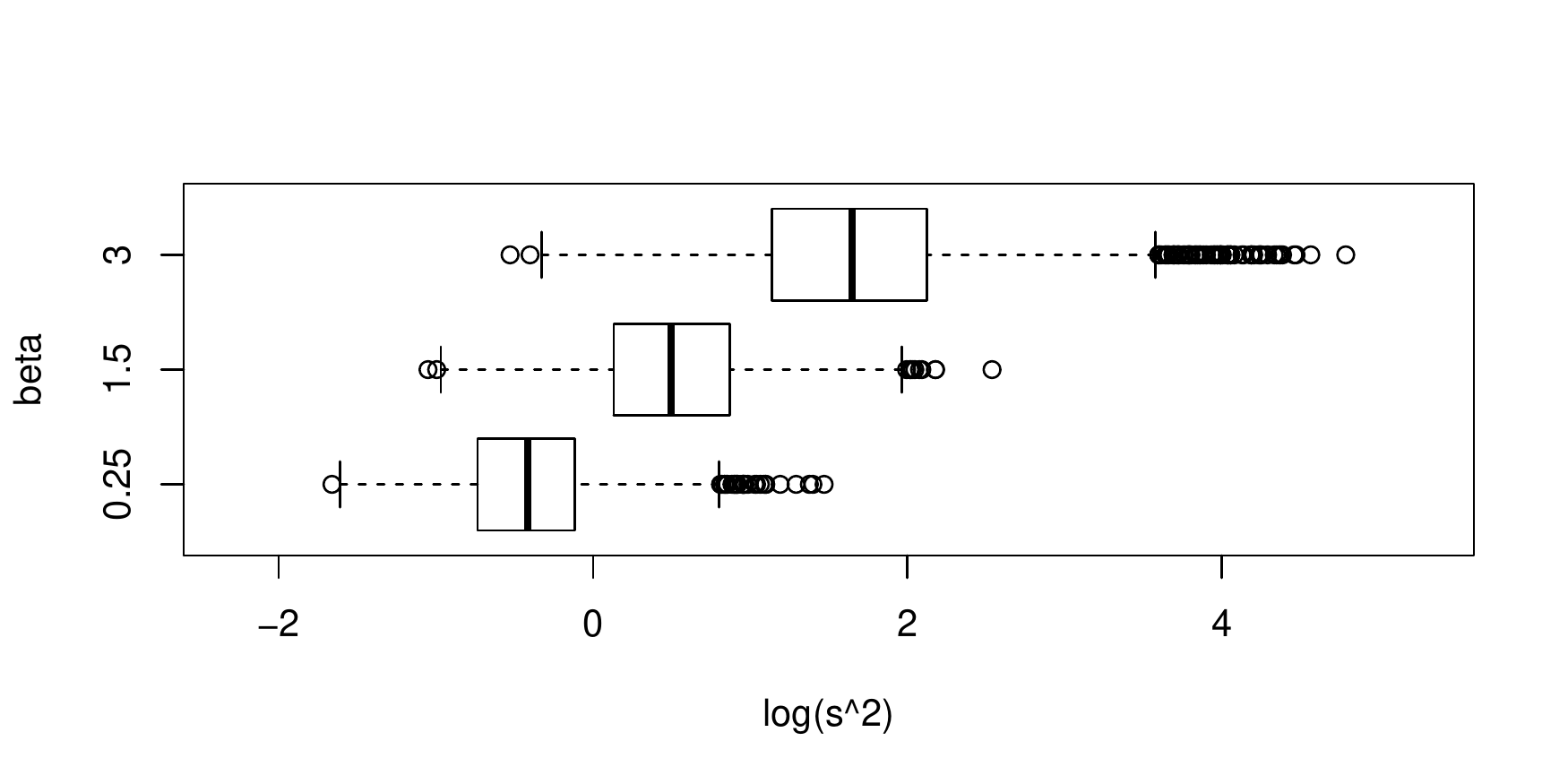}

\includegraphics[width=.3\linewidth]{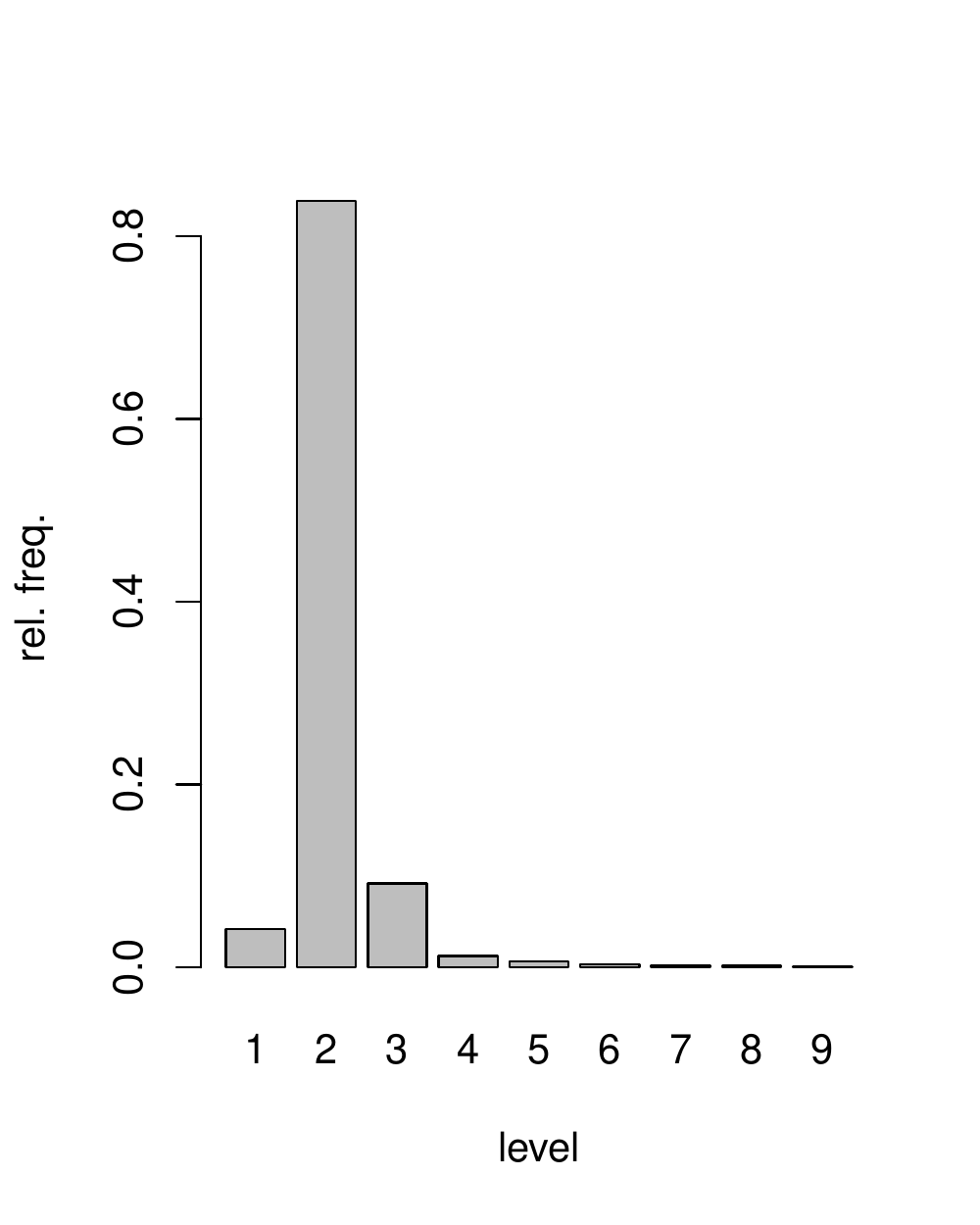}
\includegraphics[width=.3\linewidth]{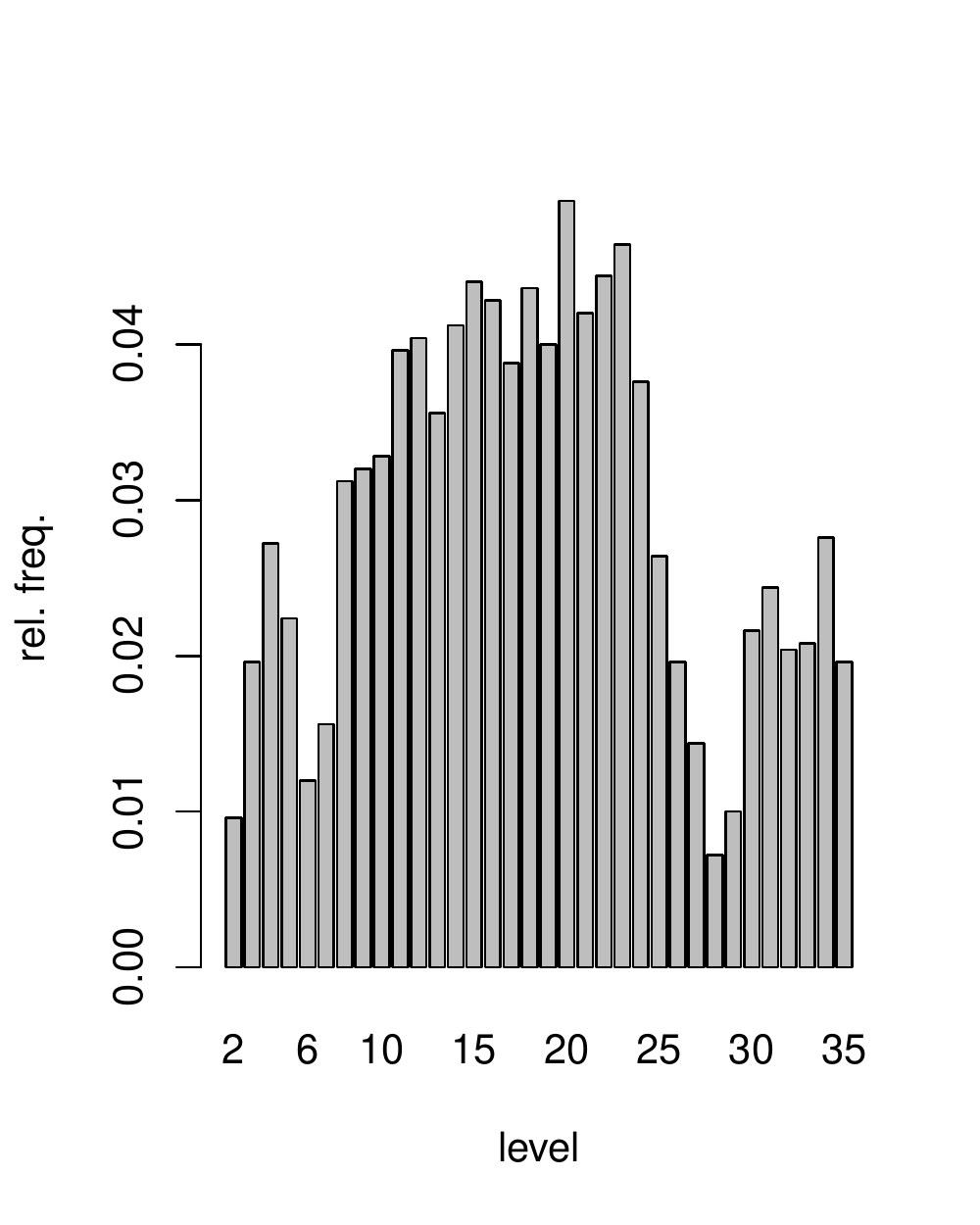}
\includegraphics[width=.3\linewidth]{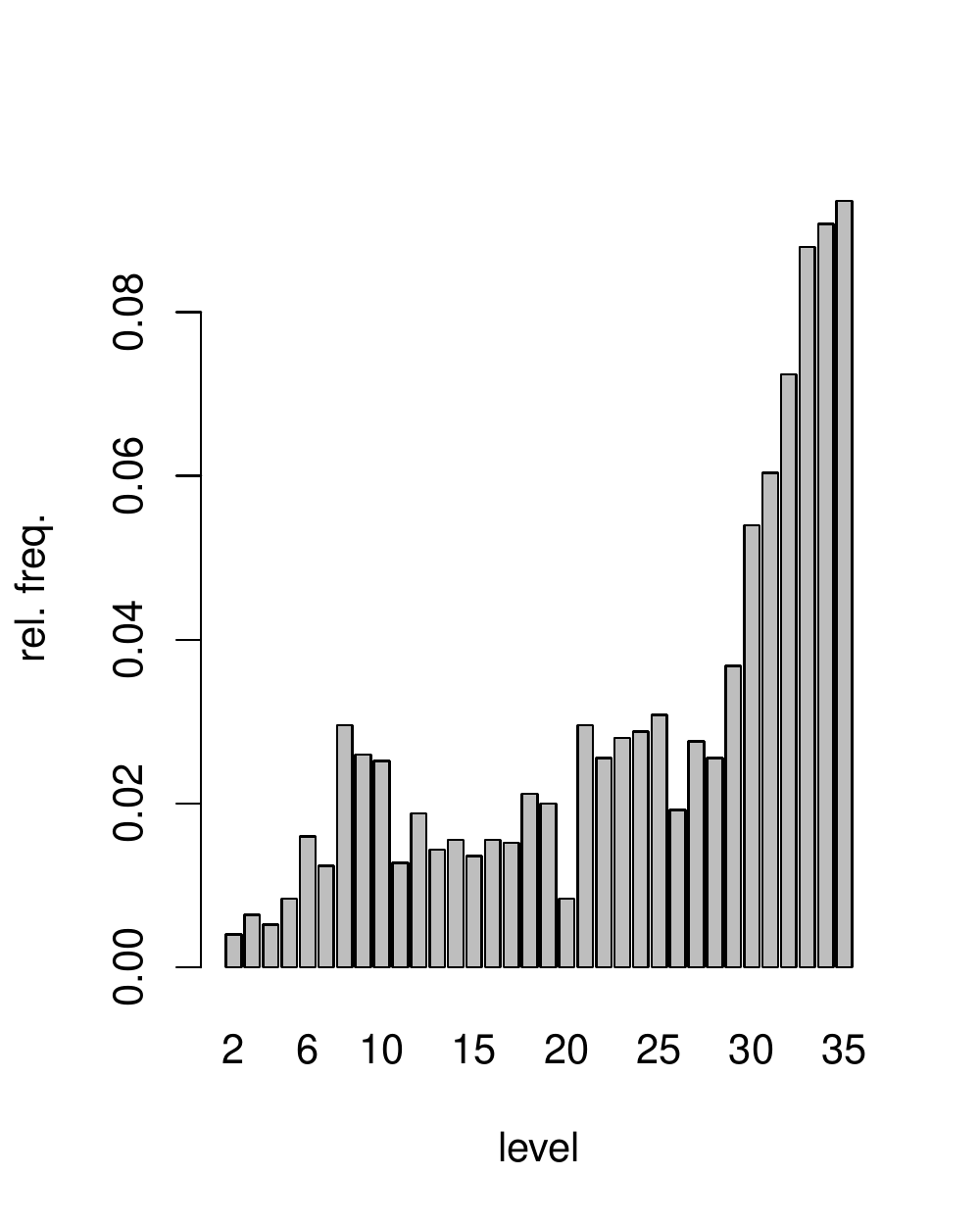}
\caption{Reversible jump; from left to right $\beta=0.25$, $\beta=1.5$ and $\beta=3$. Upper figures: posterior mean and $90 \%$ pointwise credible bands. Middle figure: boxplots of $\log(s^2)$. Lower figures: histogram of model-index. }
\label{fig:effbeta_rj}
\end{center}
\end{figure}


\begin{figure}
\begin{center}
\includegraphics[width=.3\linewidth]{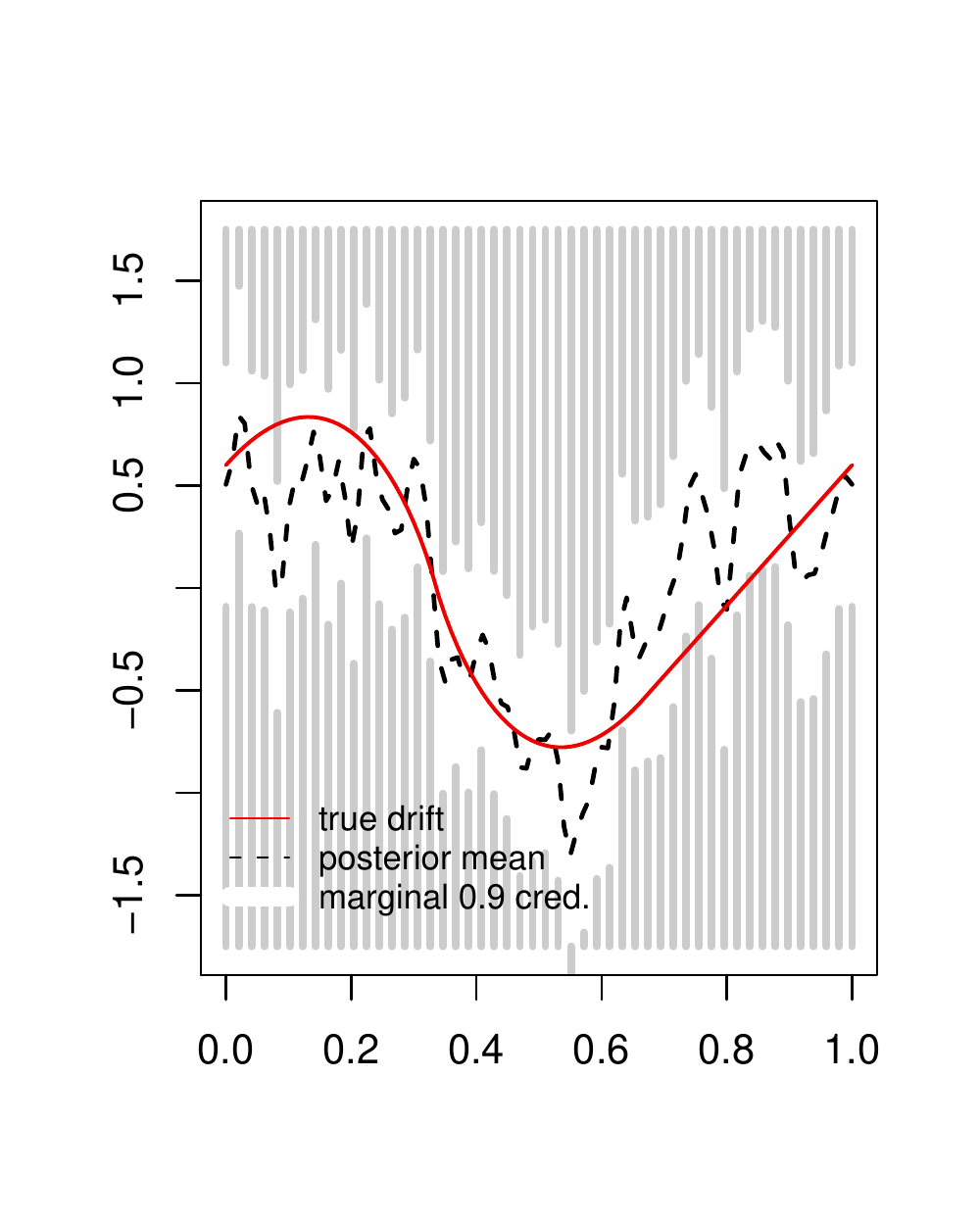}
\includegraphics[width=.3\linewidth]{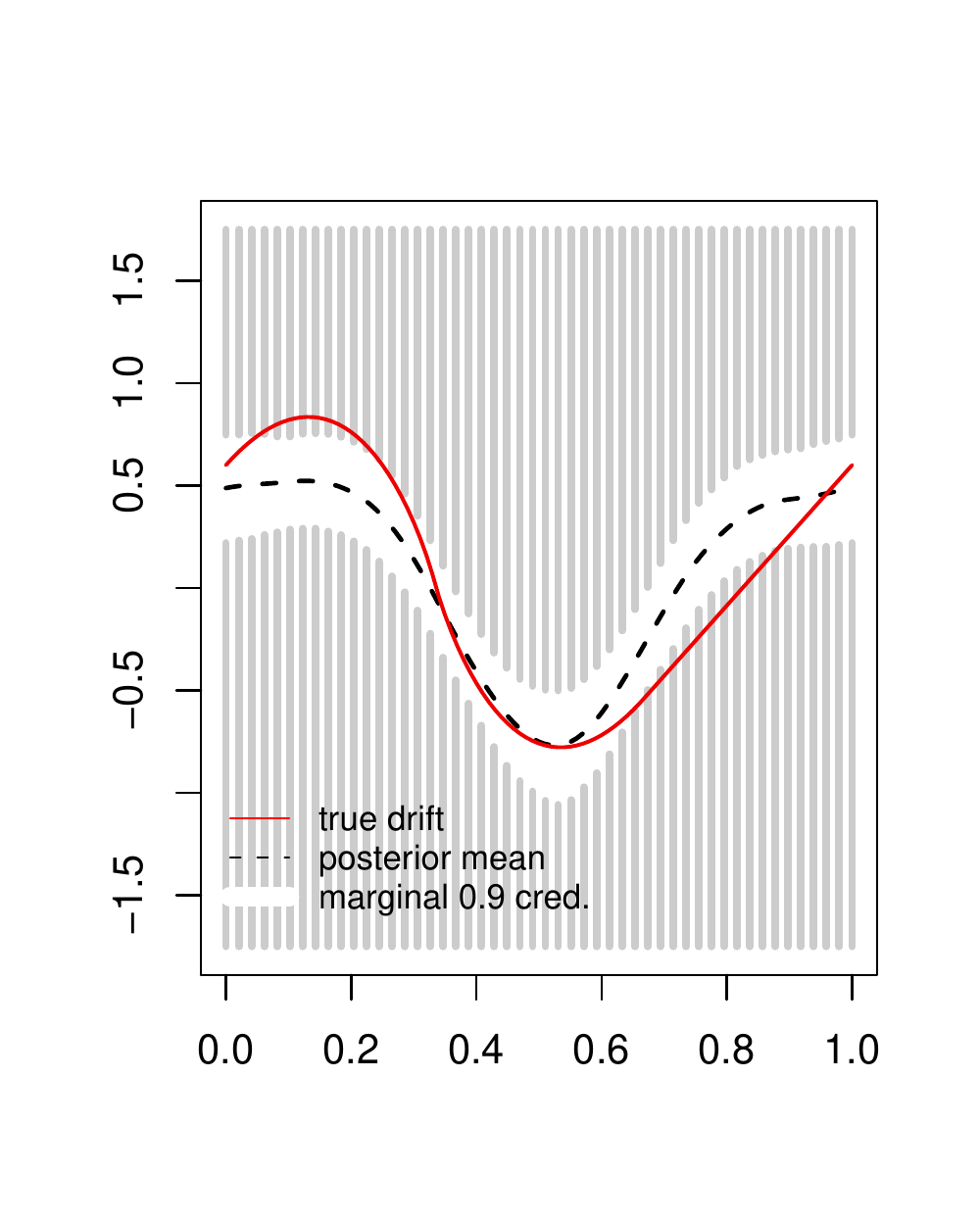}
\includegraphics[width=.3\linewidth]{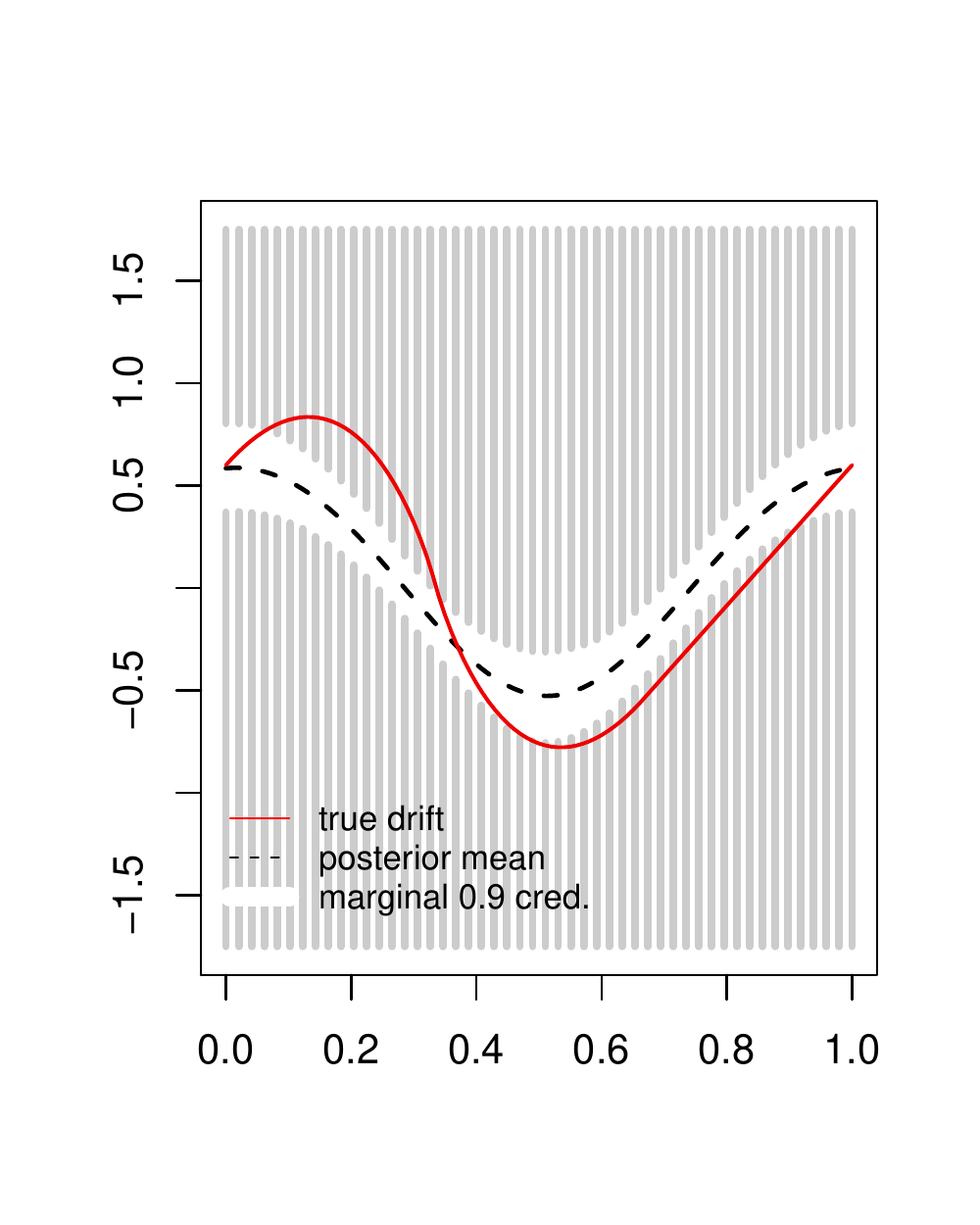}

\includegraphics[width=.8\linewidth]{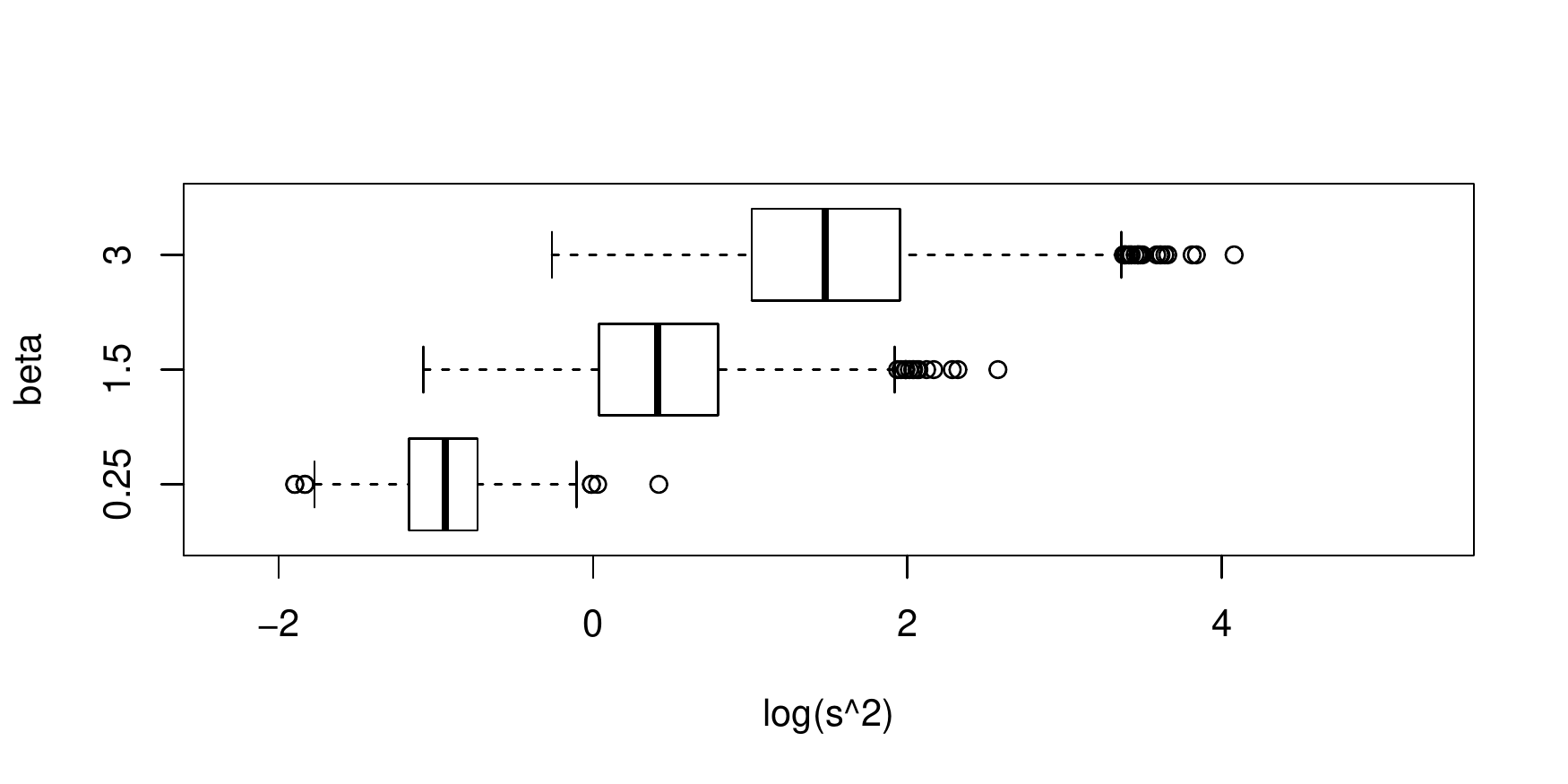}

\caption{Fixed level; from left to right $\beta=0.25$, $\beta=1.5$ and $\beta=3$. Upper figures: posterior mean and $90 \%$ pointwise credible bands. Lower figure: boxplots of $\log(s^2)$.}
\label{fig:effbeta_fixlevel}
\end{center}
\end{figure}

\subsection{Results with Schauder basis}

The complete analysis has also been done for the Schauder basis. Here we take
$q(j\mid j)=0.9$, $q(j+1 \mid j)=q(j-1\mid j)=0.1$.
The conclusions are as before. For the main example, the computing time with the Schauder basis was approximately $15\%$ of that for the Fourier basis.

\subsection{Discrete-time data and data-augmentation}
Here, we thin the ``continuous''-time data to discrete-time data by retaining every $50$th observation from the continuous-time data. The remaining data are hence at times $t=i\Delta$ with $\Delta=0.05$ and $i=0,\ldots, 4000$.  Next, we use our algorithm both with and without data-augmentation. Here, we used the Schauder basis to reduce  computing time.  The results are depicted in Figure \ref{fig:effect_da}.
\begin{figure}
\begin{center}
\includegraphics[width=.3\linewidth]{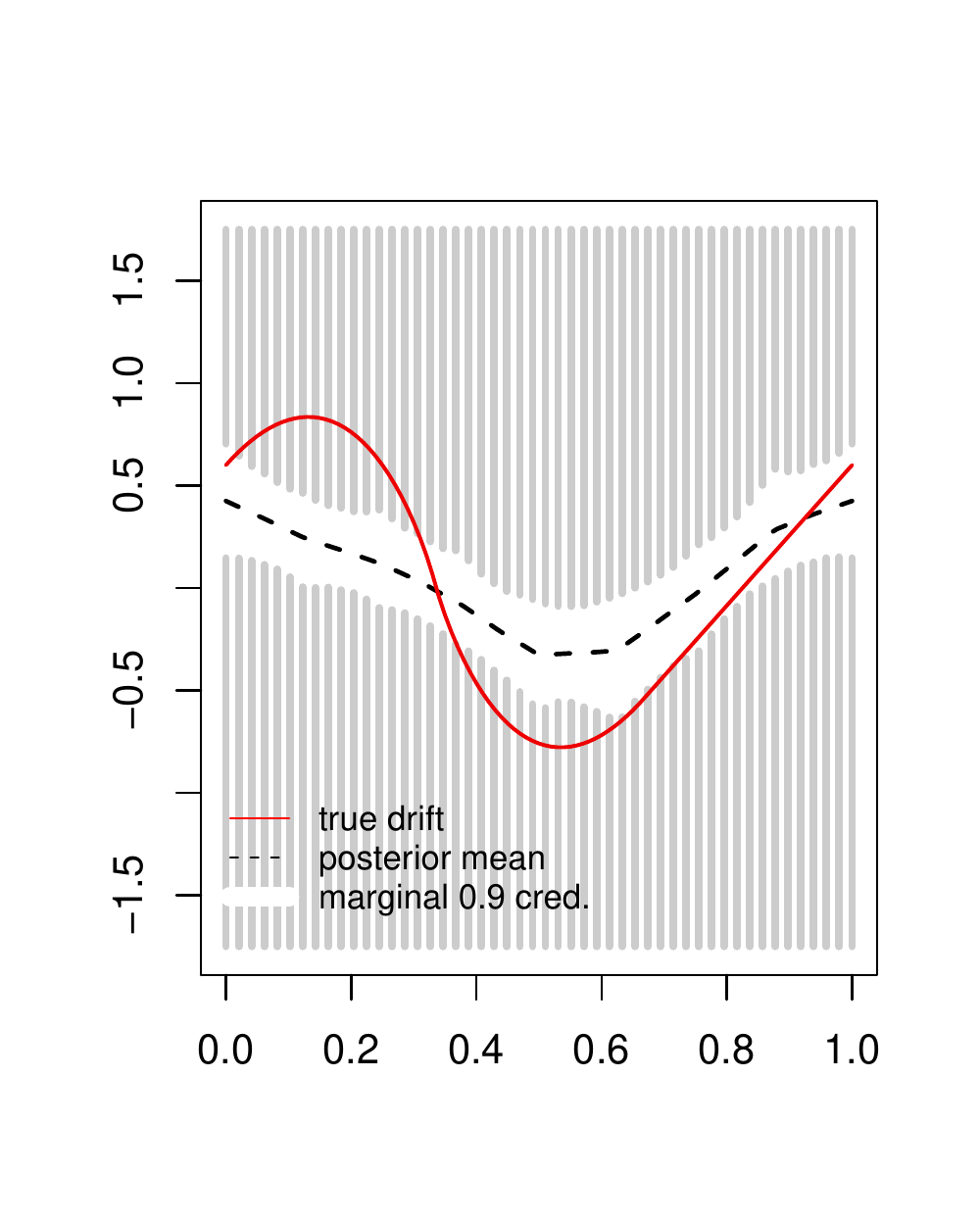}
\includegraphics[width=.3\linewidth]{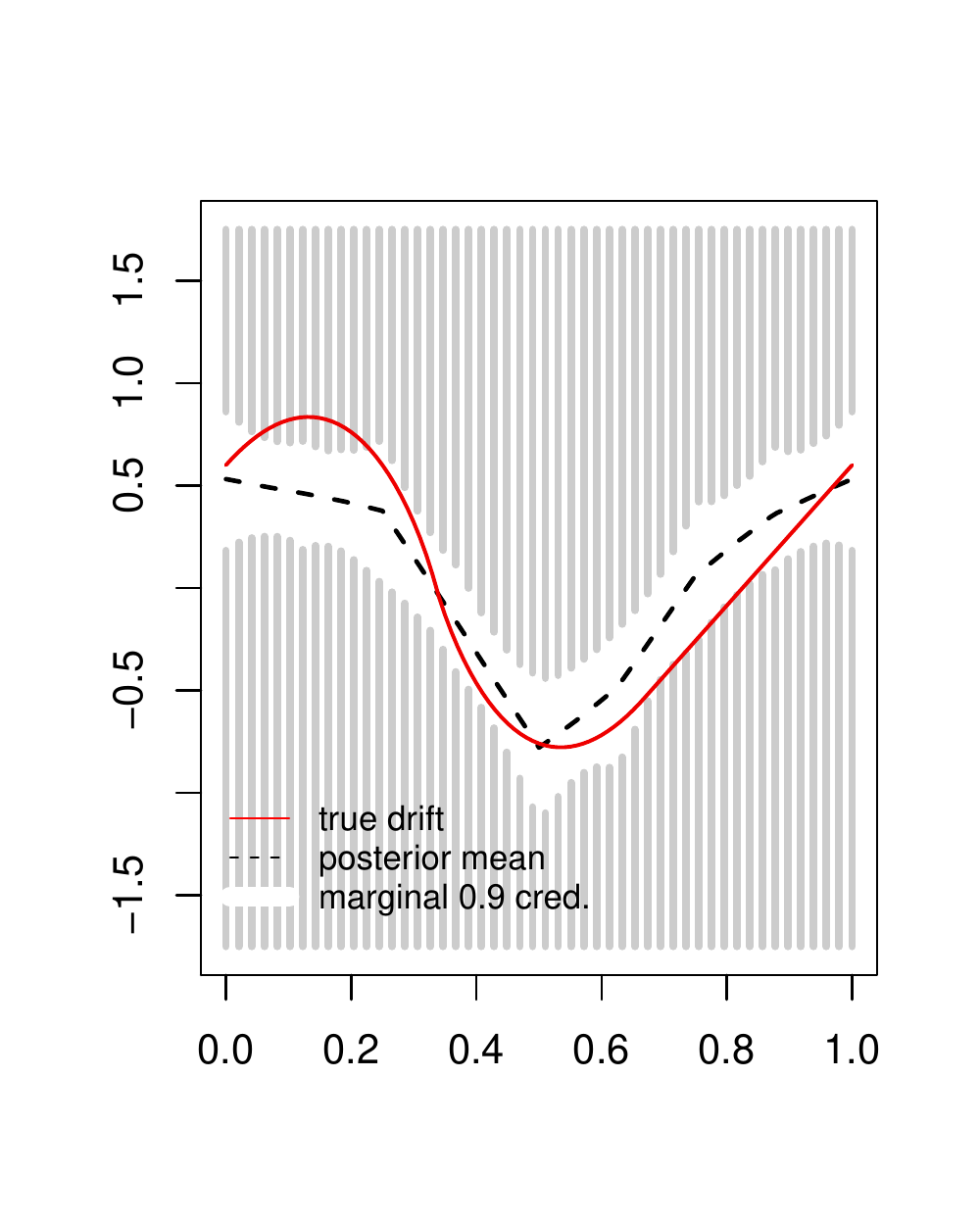}
\includegraphics[width=.3\linewidth]{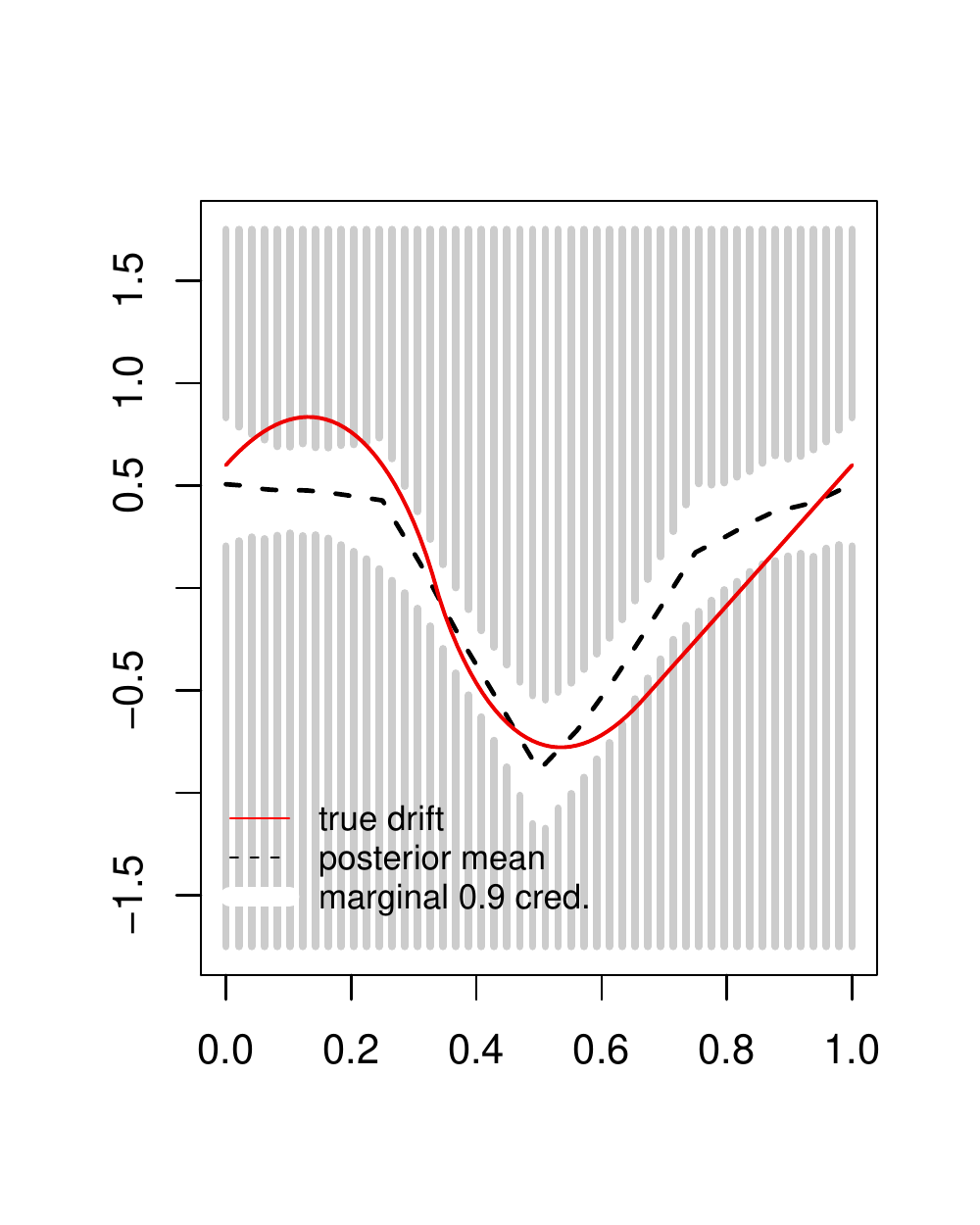}
\caption{Drift function (red, solid), posterior mean (black, dashed) and $90 \%$ pointwise credible bands. Discrete-time data.  Left: without data-augmentation.  Middle: with data-augmentation. Right: continuous-time data.}
\label{fig:effect_da}
\end{center}
\end{figure}

The leftmost plot clearly illustrates that the discrete-time data with $\Delta=0.5$ cannot be treated as continuous-time data. The bias is quite apparent. Comparing the middle and the rightmost plot shows that data-augmentation works very well in this example.

We also looked at the effect of varying $T$ (observation time horizon) and $\Delta$ (time in between discrete time observations).
In all plots of Figure \ref{fig:vary} we used the Schauder basis with $\beta=1.5$ and data augmentation with $49$ extra imputed points in between two observations.
In the upper plots of Figure \ref{fig:vary} we varied the observation time horizon while keeping $\Delta=0.1$ fixed. In the lower plots of Figure \ref{fig:vary} we fixed $T=500$ and varied $\Delta$. As expected, we see that the size of the credible
bands decreases as the amount of information in the data grows.

Lastly, Figure \ref{fig:vary3} illustrates the influence of increasing the number of augmented observations on the mixing of the chain.
Here we took $\Delta = 0.2$ and $T=500$ and compare trace plots for two different choices of the number of augmented data points, in one case $25$ data points per observation and in the second case $100$ data points per observations. The mixing does not seem to deteriorate with a higher number of augmented observations.

\begin{figure}
\begin{center}
\includegraphics[width=.25\linewidth]{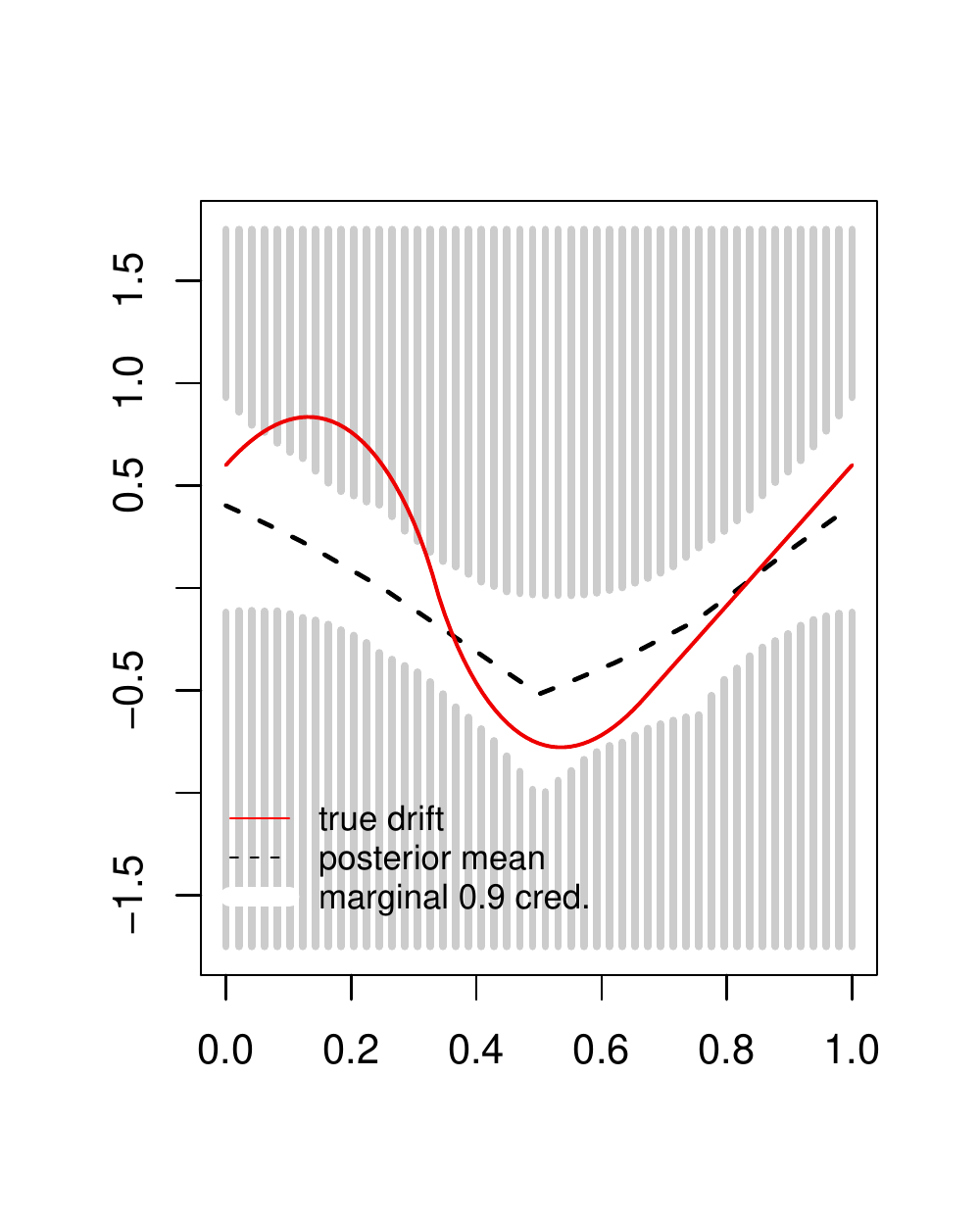}
\includegraphics[width=.25\linewidth]{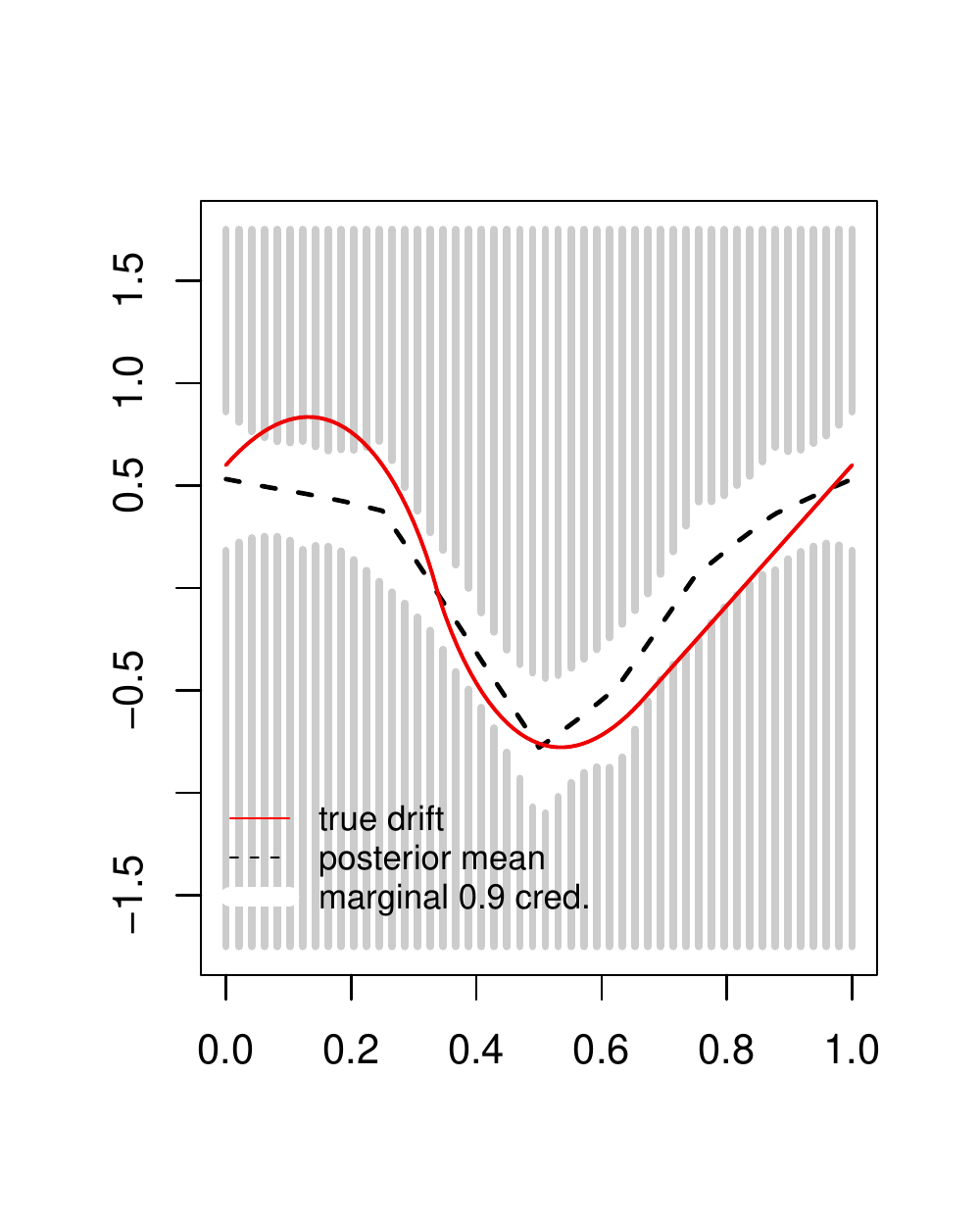}
\includegraphics[width=.25\linewidth]{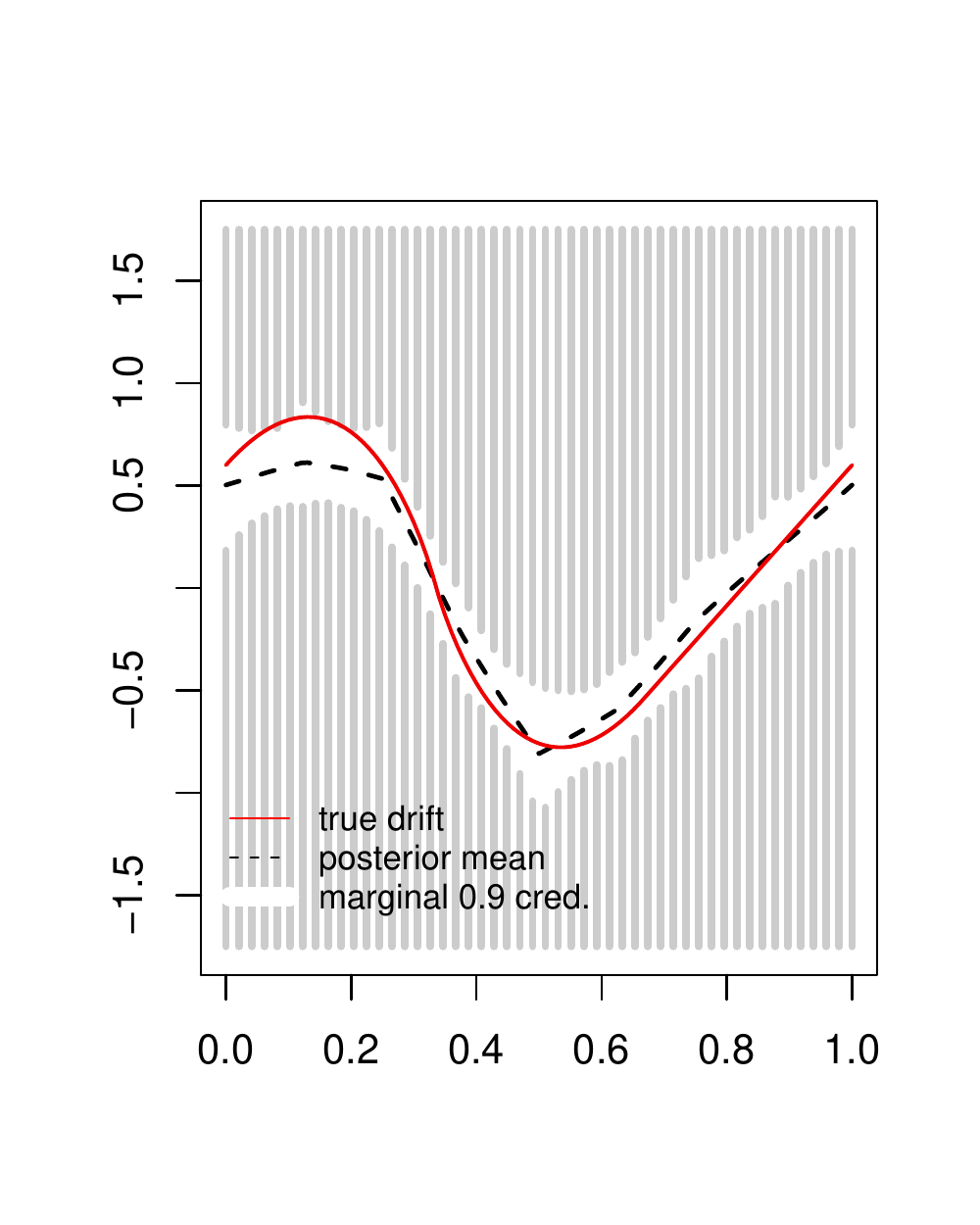}

\includegraphics[width=.25\linewidth]{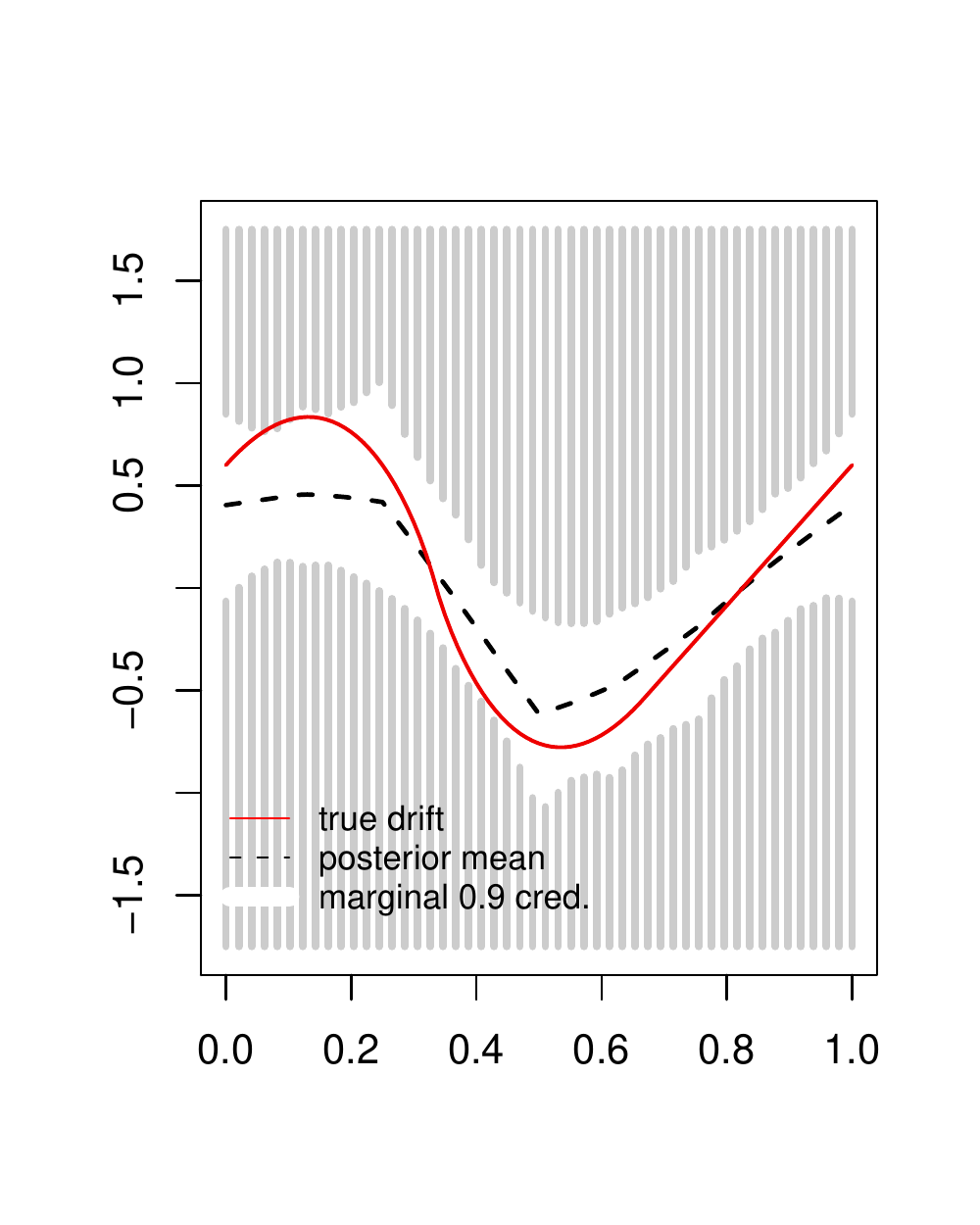}\includegraphics[width=.25\linewidth]{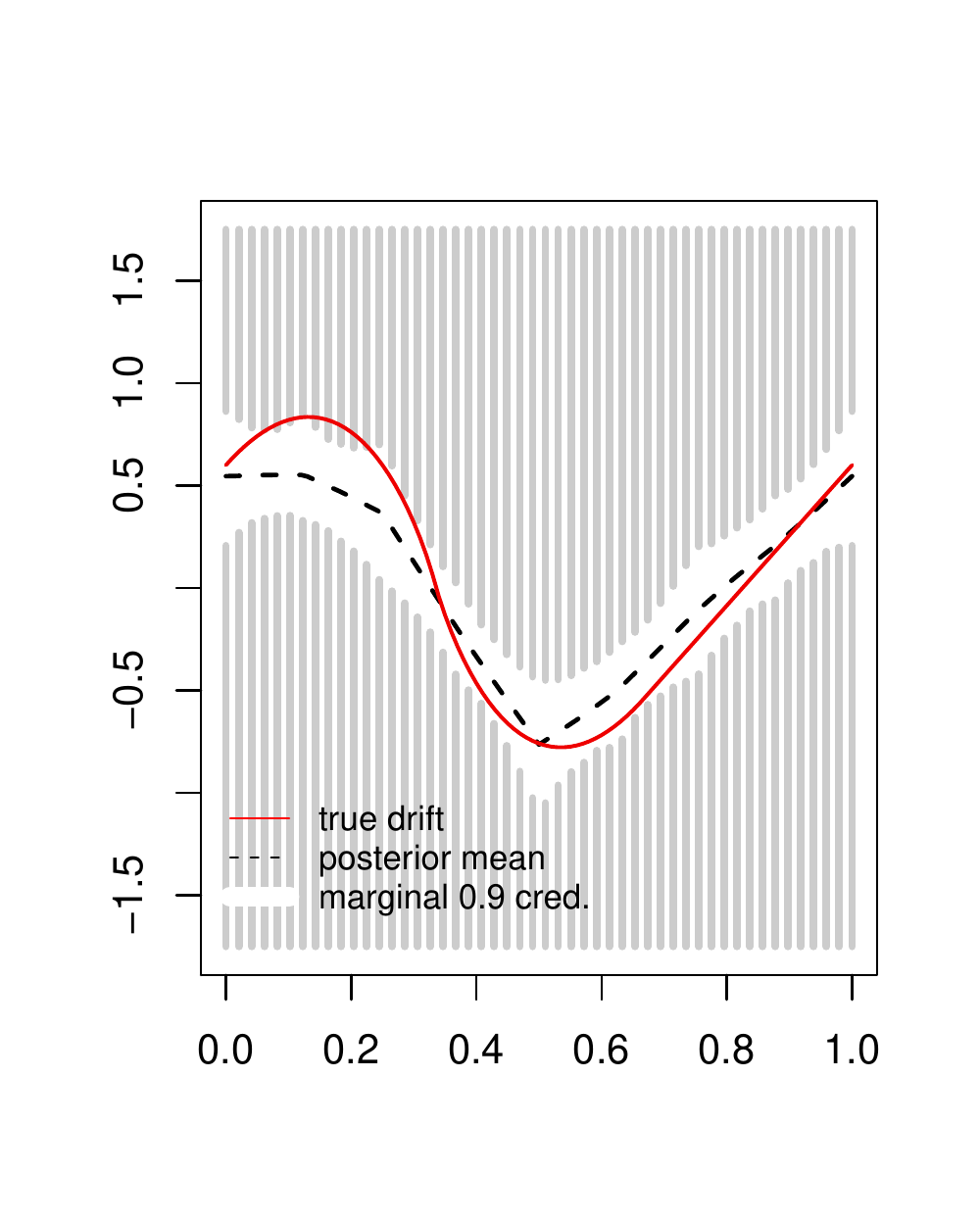}
\includegraphics[width=.25\linewidth]{images/bb-sch-1,5-periodic-aug-493000final5000in500post.pdf}\\

\caption{Drift function (red, solid), posterior mean (black, dashed) and $90 \%$ pointwise credible bands.
Upper: Discrete time observations with $\Delta=0.1$. From left to right $T=50$, $T=200$ and $T=500$. Lower: Discrete time observations with $T=500$. From left to right $\Delta=1$, $\Delta=0.2$ and $\Delta=0.1$. (All augmented to $\delta = 0.002$.) }
\label{fig:vary}
\end{center}
\end{figure}

\begin{figure}
\begin{center}
\includegraphics[width=.80\linewidth]{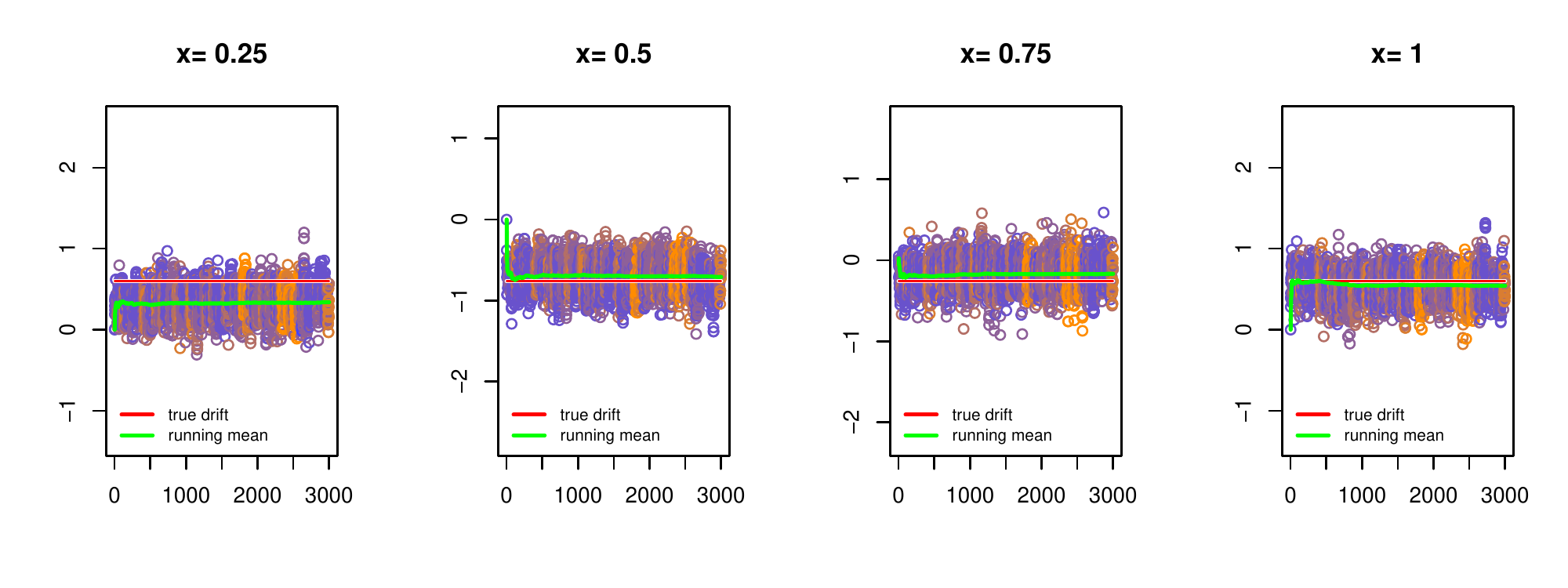}
\includegraphics[width=.80\linewidth]{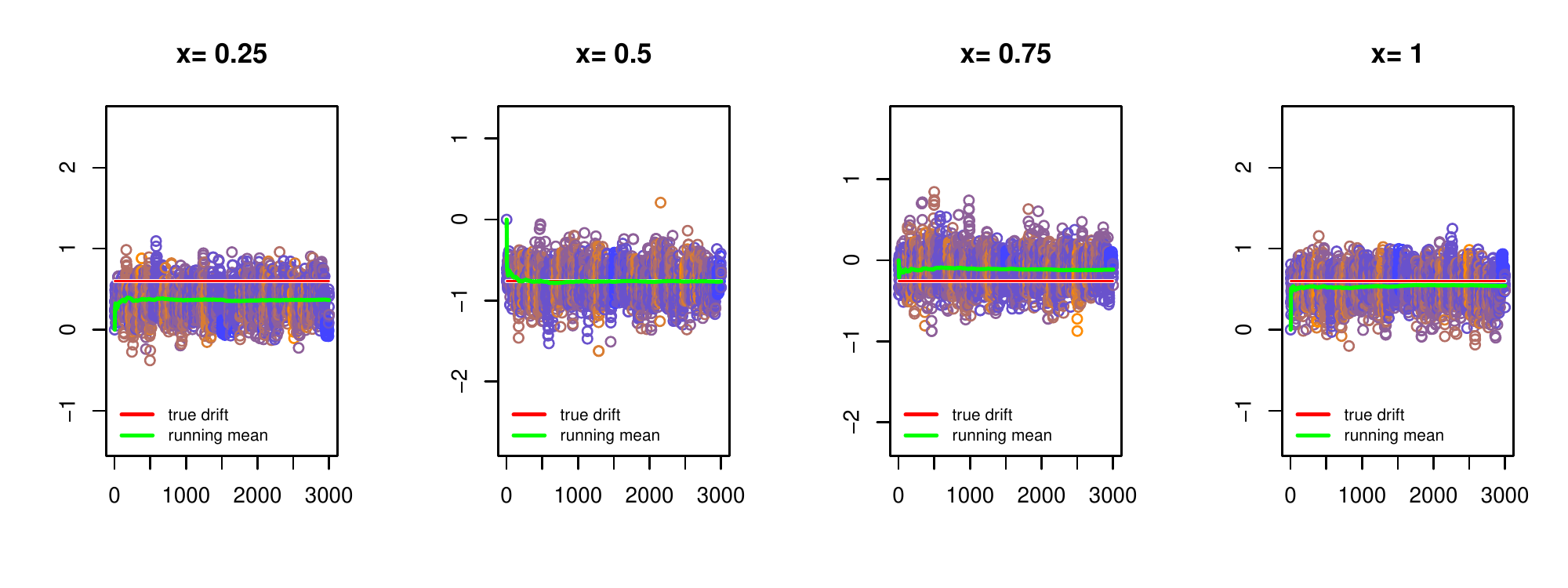}
\caption{Trace plots illustrating the influence data augmentation on the mixing of the chain. 2500 observations in [0,500]. Top: $25$ data points per observation. Bottom: $100$ data points per observation.}\label{fig:vary3}
\end{center}
\end{figure}

\subsection{Performance of the method for various drift functions}\label{sec:gallery}
In this section we investigate how different features of the drift function influence the results of our method.
The drift functions chosen for the comparison are
\begin{enumerate}
\item
$b_1(x) =  8\sin(4\pi x)$,
\item
$b_2(x) = 200 \tilde x (1-2\tilde x)^3 \ind{[0,\frac12)} (\tilde x)  - \frac{400}{3} (1-\tilde x)(2\tilde x-1)^3\ind{[\frac12,1)}(\tilde x)$, where $\tilde x = x {\mod{}} 1$,
\item
$b_3(x) = -8\sin(\pi (4x-1)) \ind{[\frac14, \frac34]}(x {\mod{}} 1)$.
\end{enumerate}
As a prior we took the Fourier basis with parameter $\beta=1.5$. For $s^2$ we took an inverse Gamma prior with hyper parameters $a = b = 5/2$.  For each drift function, $2000$ observations with  $\Delta = 0.1$ were generated. The algorithm was then used with data augmentation with $49$ imputed points extra in between two observations. The results for $b_1$, $b_2$ and $b_3$ are in figure
\ref{fig:gallery}. We ran the analysis for the Schauder basis as well, which led to very similar results.

\begin{figure}
\begin{center}
\includegraphics[width=.3\linewidth]{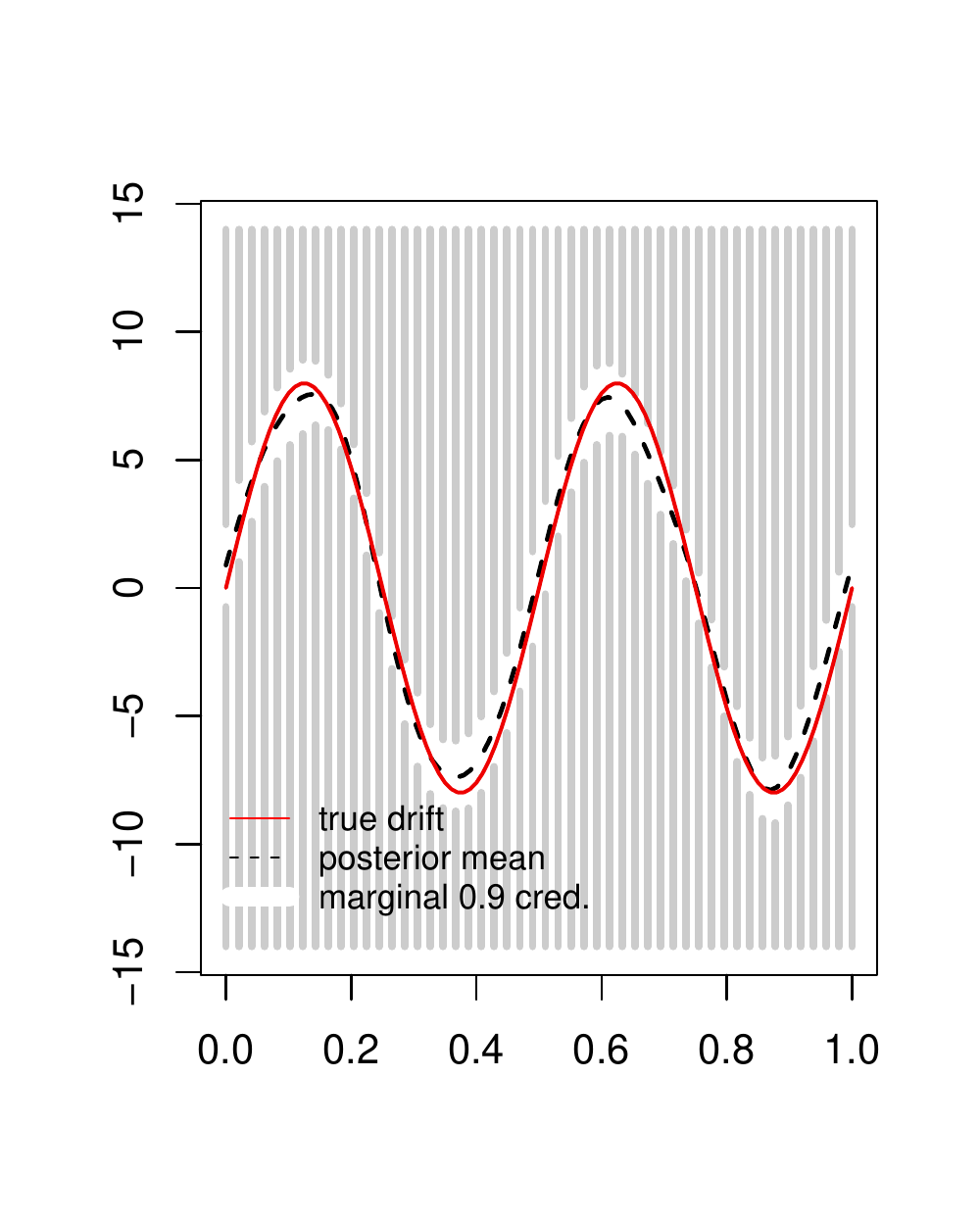}
\includegraphics[width=.3\linewidth]{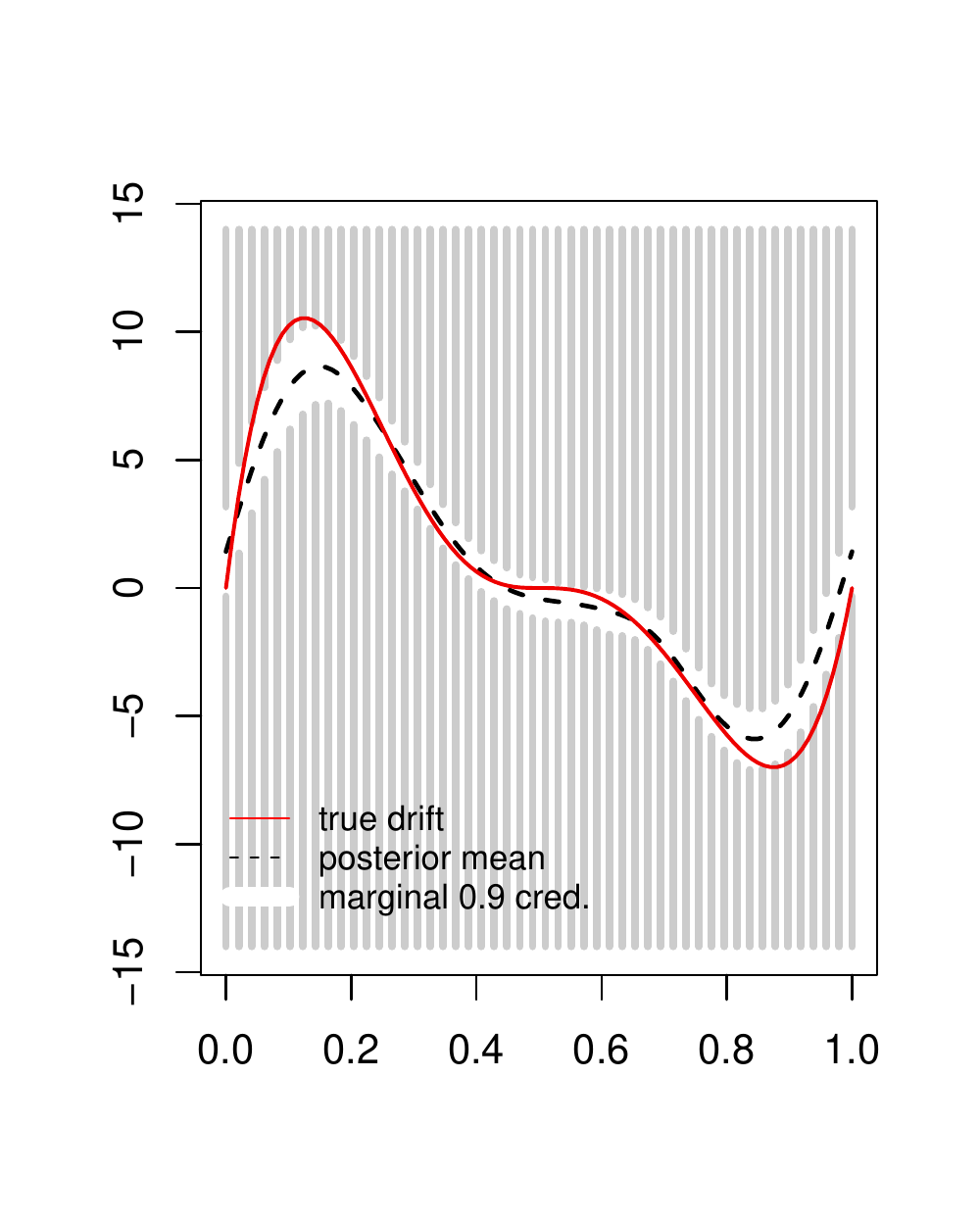}
\includegraphics[width=.3\linewidth]{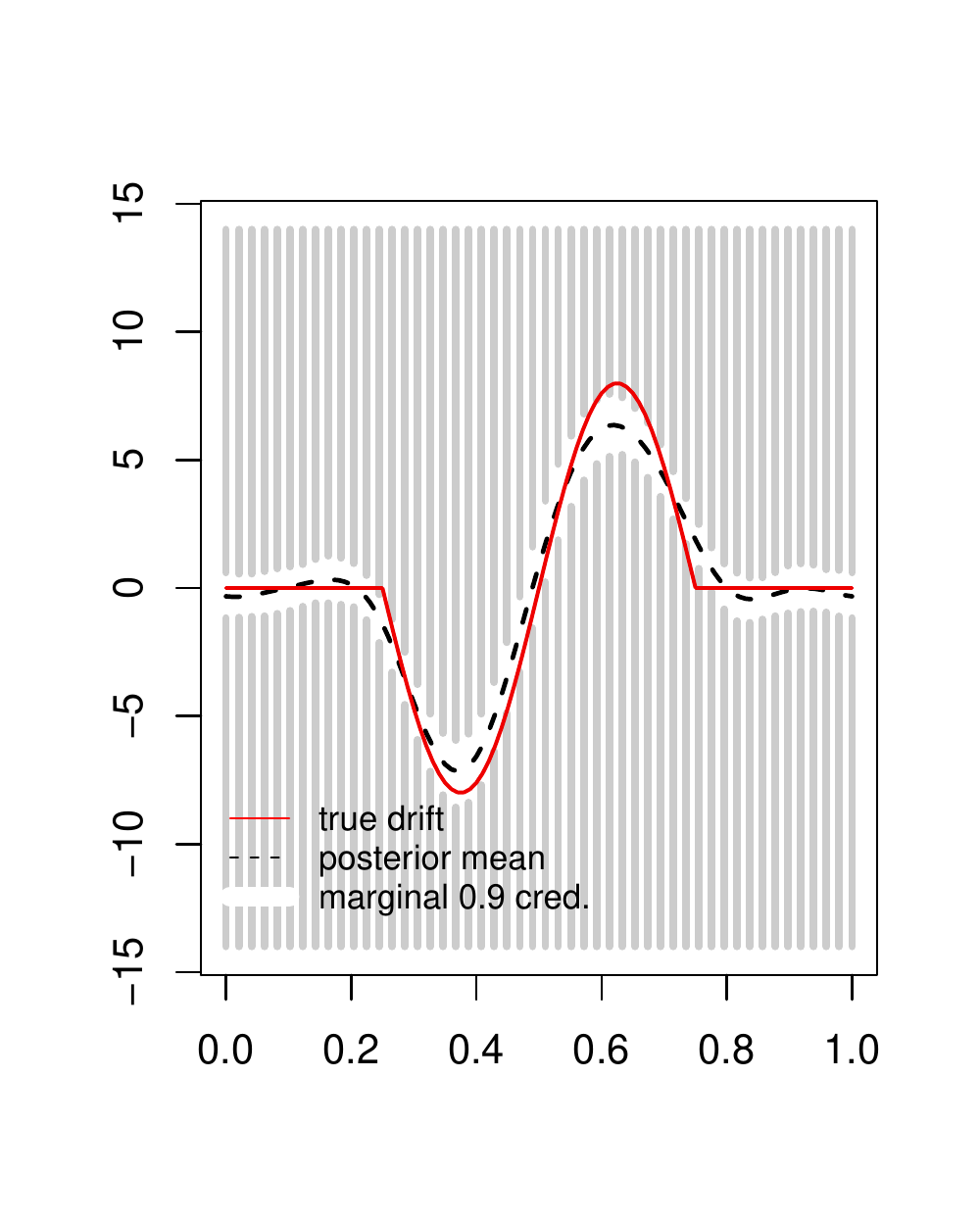}
\caption{Drift function (red, solid), posterior mean (black, dashed) and $90 \%$ pointwise credible bands. From left to right: $b_1$, $b_2$, and $b_3$. }\label{fig:gallery}
\end{center}
\end{figure}

\subsection{Comparison with \cite{Pokern}: Butane Dihedral Angle Time Series}

To compare our approach to that of \cite{Pokern} we analyzed the butane dihedral angle time series considered by these authors. After some preliminary operations on the data, these data are assumed to be discrete time observations from a scalar diffusion with unit-diffusion coefficient (details on this are described in \cite{Pokern} and supplementary material to this article).
After this preliminary step,  the time series consists of $4000$ observations observed evenly over the time  interval $[0,4]$ (time is measured in nanoseconds). The right-hand figure of  Figure \ref{fig:pokern} shows a histogram of the discrete time observations.


\cite{Pokern} use a centered Gaussian process prior with precision operator \eqref{eq: *},  with $\eta=0.02$, $\kappa=0$ and $p=2$. This choice for $p$ yields a prior of H\"older smoothness essentially equal to $1.5$.
As explained in the introduction, this is essentially  the law of the random function
\[
x \mapsto \sum_{l=1}^\infty \frac1{l^2\pi^2 \sqrt{\eta}}  Z_l \psi_l(x),
\]
where $\psi_l$ are the Fourier basis functions defined in Section \ref{sec: fourier} and the $Z_l$ are
independent standard normal random variables. Note that with $\eta=0.02$ we have $(\pi^4 \eta)^{-1}\approx 0.51$.
To match as closely as possible with their prior specification, we use the Fourier basis with $\beta=1.5$, as described in Section
\ref{sec: fourier}. Conditionally on $s^2$ and $j$, our prior then equals the law of the random function
\[ x \mapsto \sum_{l =1}^{2j-1} \frac{s}{l^2} Z_l \psi_l(x). \]
For $s^2$ we took an inverse Gamma prior with hyper parameters $a = b = 5/2$. For this choice
the prior mean for $s^2$ equals $0.5$ which is close to $(\pi^4 \eta)^{-1}$.

 For the prior on the models we use $p(j)\propto (0.95)^{m_j}$. As before, we ran the continuous time algorithm for $3000$ cycles and discarded the first $500$ iterations as burn-in. We used data augmentation with $99$ extra augmented time points in between two successive observations.


The left-hand figure in  Figure \ref{fig:pokern} shows the posterior mean and pointwise $68\%$ credible bands, both for our approach and that of \cite{Pokern}. Overall, the posterior means computed by both methods seem to agree very well except for the boundary areas. In these areas, the credible bounds are wider, since we have less information about the drift here.
In Figure \ref{fig:pokern2} histograms of the scaling parameter $s^2$ and the truncation level $J$ are shown. Clearly, $s^2$ takes values typically around a value as large as $5000$, much larger than $0.51$. This illustrates once again the usefulness of equipping the scaling parameter with a prior distribution.
The fact that our credible bands are wider near the boundary of the observation area seems to indicate that
\cite{Pokern} are somewhat overconfident about the form of the drift function in that area.
Their narrower credible bands  seem to be  caused by prior belief rather than information in the data and
are not corroborated by our more conservative approach.

\begin{figure}
\begin{center}
\includegraphics[width=.5\linewidth]{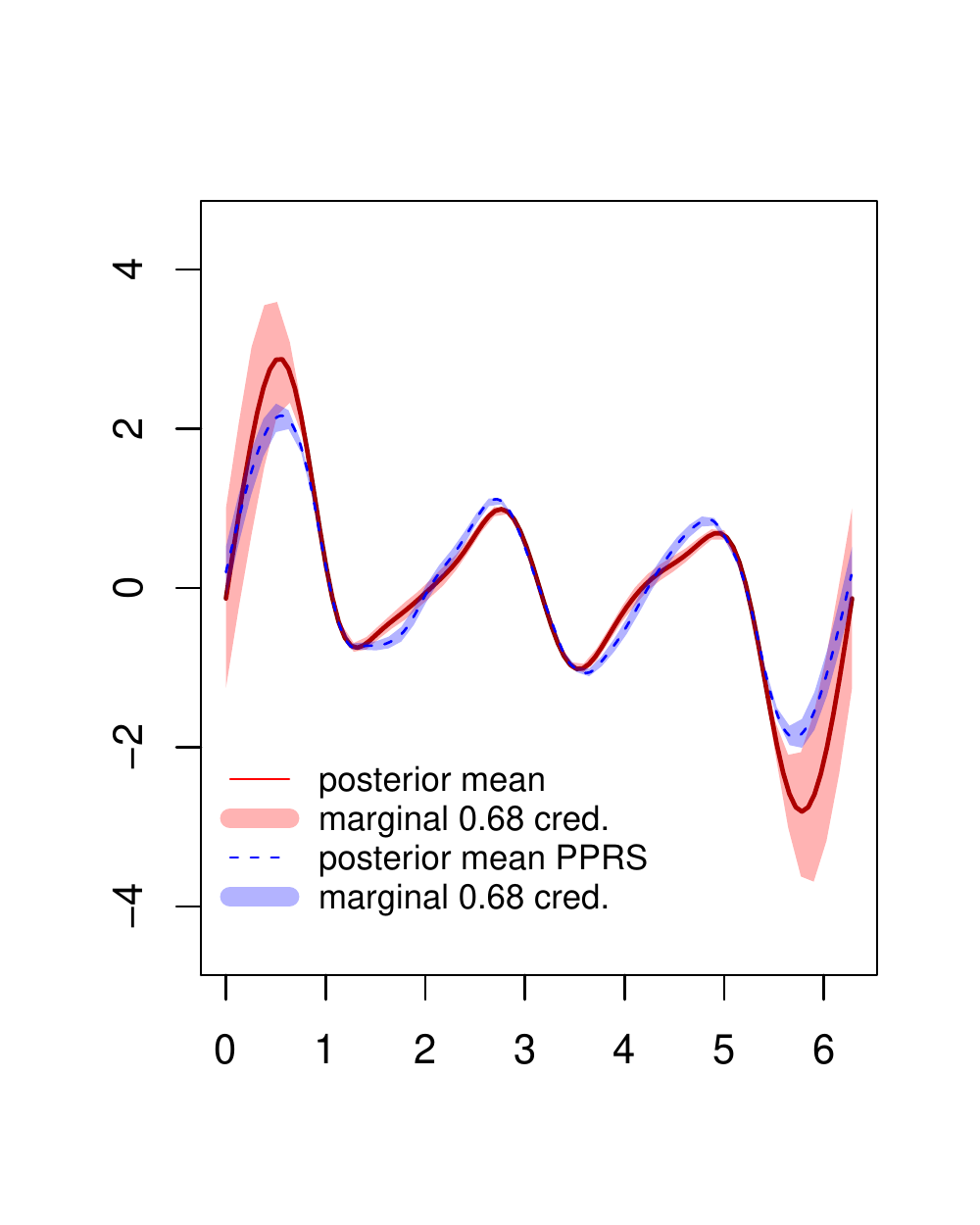}\includegraphics[width=.5\linewidth]{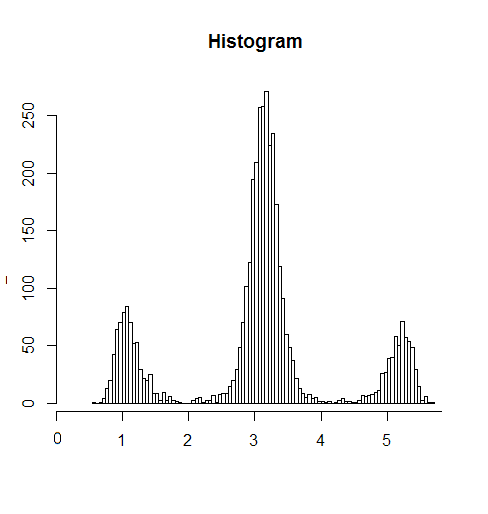}
\caption{Comparison of the estimate of drift using the Butane Dihedral Angle data. Red solid: A Fourier prior with $\beta=1.5$. Blue dashed: Results of \cite{Pokern}. The posterior mean with 68\% credible bands is pictured.  Right: Histogram of the data.  }\label{fig:pokern}
\end{center}
\end{figure}

\begin{figure}
\begin{center}
\includegraphics[width=\linewidth]{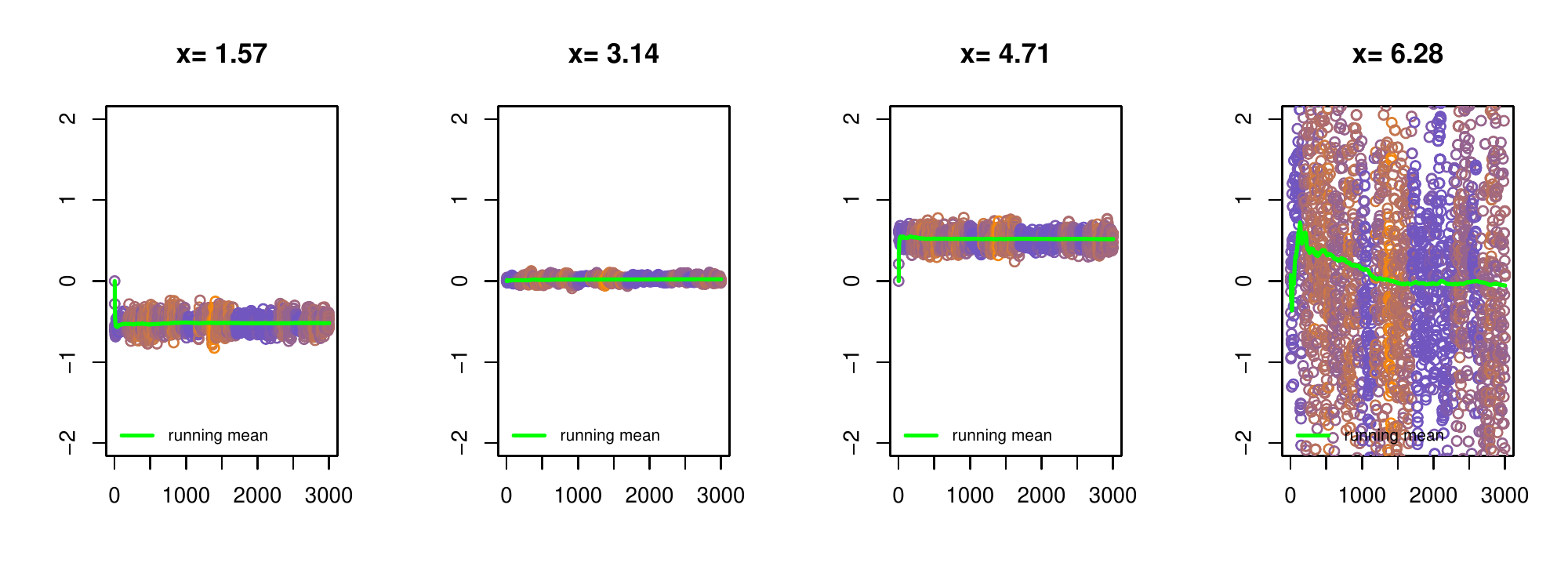}\\
\includegraphics[width=.45\linewidth]{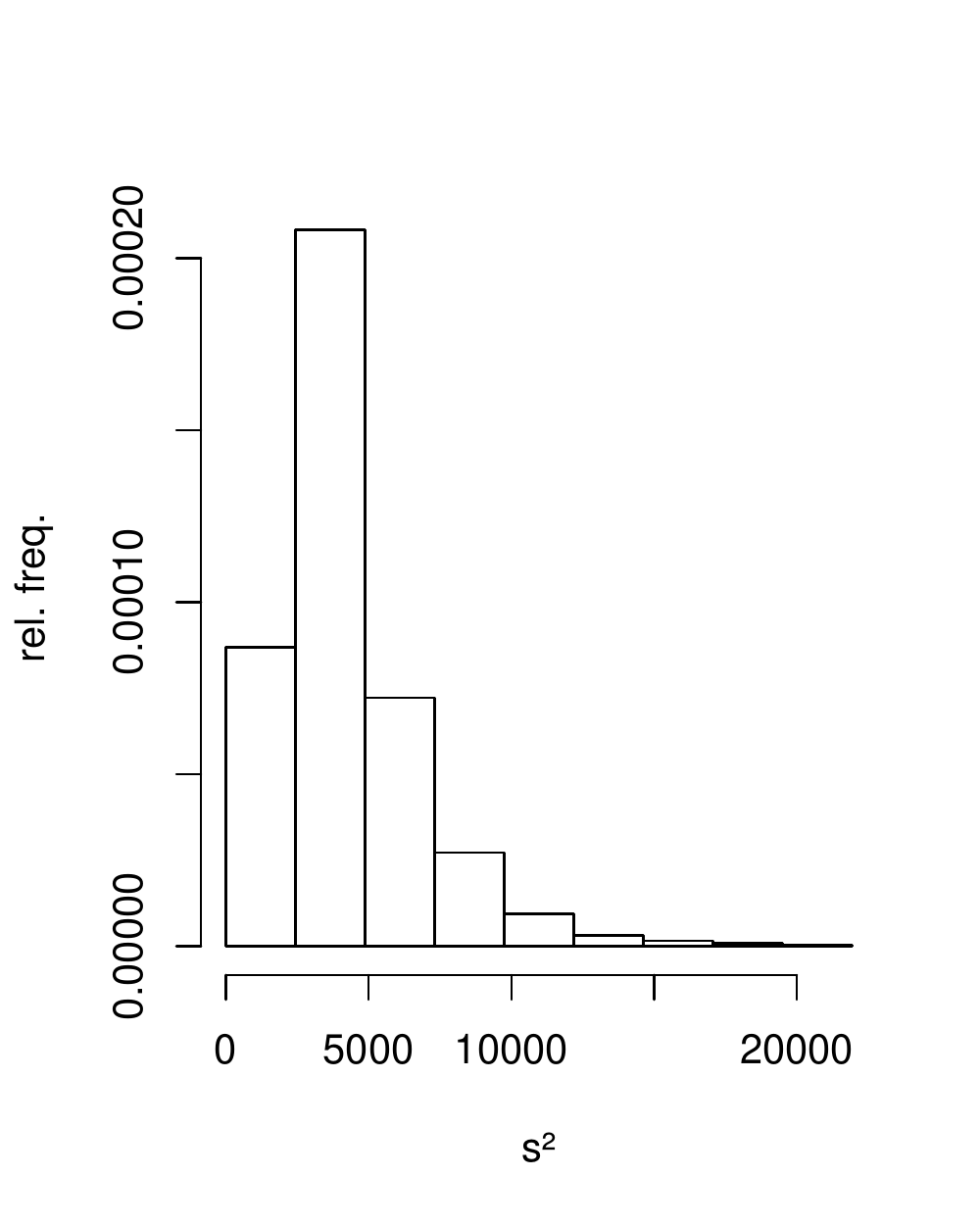}\includegraphics[width=.45\linewidth]{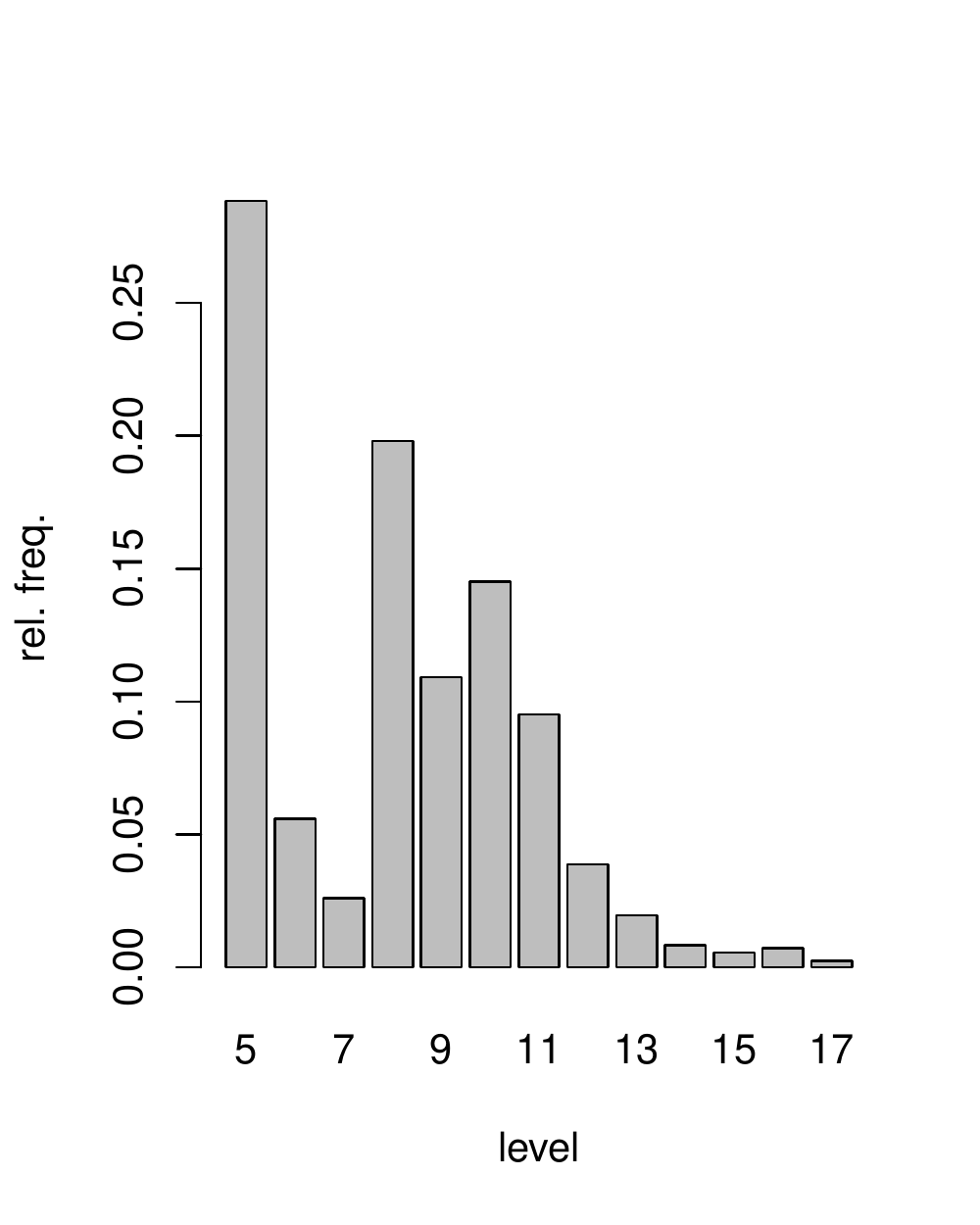}
\caption{Trace plot, histogram of $j$, histogram of $s^2$.}\label{fig:pokern2}
\end{center}
\end{figure}

\section{Numerical considerations}
\label{sec: num}

\subsection{Drawing coefficients from the posterior within a fixed model}\label{posteriorwithin}

For move II the algorithm requires to   sample a random vector  $U \sim N_{m_j}((W^{j})^{-1}\mu^{j}, (W^{j})^{-1})$.
In order to do so  we first compute the Cholesky decomposition of   $W^{j}$ (note that $W^{j}$ is symmetric and positive definite,
ensuring its existence). For an upper triangular matrix $M^j$ we then have
$W^{j} = (M^{j})^T M^{j}$.  Next we let $z^j$ solve the system $(M^j)^T z^j=\mu^j$,  draw
 a standard normal vector $Z \sim {\mr N}_{m_j}(0,I)$ and construct  $U$ by backward solving
\begin{equation}
M^j U =z^j + Z.
\label{eq:backsolve}
\end{equation}
It is easily seen that the random vector $U$ has the required distribution.

Backsolving linear equations with triangular matrices requires $\mc{O}(m_j^2)$ operations. Cholesky factors are computed in $\mc{O}(m_j^3)$ operations in general, but for basis functions  with local support, $\Sigma^j$ and $W^j$ are sparse, enabling enabling faster computations. For the Schauder basis, the number of non-zero elements of the upper triangular part of $\Sigma^j$ is $2^{j-1}(j-1) + 1$, so the fraction of non-zero elements of $\Sigma^j$ is approximately $1.00, 1.00, 0.88, 0.66, 0.45, 0.28, 0.17, \dots$ for $j=1,2,3,4,5, 6, 7, \dots$
The Cholesky factor of a sparse matrix is not necessarily sparse as well. However,  the sparsity pattern originating from the tree structure of the supports of the Schauder elements enables to specify a \emph{perfect elimination ordering} of the rows and columns of $\Sigma^j$ (for details we refer to \cite{Rose}). This means  that the Cholesky factor inherits the sparsity. Moreover,  the Cholesky factorization can be computed on the sparse representation of the matrix $\Sigma^j$. The particular reordering necessary -- reversing the order of rows and columns -- makes this technique applicable for moves within levels.

\subsection{Computation of the Bayes factors}

Our algorithms  require the evaluation of the Bayes factors defined by (\ref{eq:rel}).
The following lemma is instrumental in the numerical evaluation of these numbers.
Recall the definitions of $\mu^j$, $\Sigma^j$ and $W^j$ in Section \ref{sec: within}.

\begin{lem}\label{lem:predictive_density}
We have
$$	
p(x^T \given j, s^2) = \frac{ \exp\left(\tfrac12(\mu^j)^T( W^j)^{-1} \mu^j\right)}{\sqrt{ |s^2 W^j \Xi^j |}}.
$$
\end{lem}

\begin{proof}
Since
$$	p(x^T \given j, s^2) =\int p(x^T \given j, \th^j,  s^2) p(\th^j \given j,  s^2) d \th^j$$
we have, by  (\ref{eq: gg}) and the definition of the prior,
\begin{equation}
\label{eq:pred_calc}
 p(x^T \given j, s^2) = \frac{1}{\sqrt{|2\pi s^2\Xi^j|}}
 \int e^{(\theta^j)^T\mu^j -\frac12(\theta^j)^TW^j \theta^j}  d \th^j.
\end{equation}
By completing the square we see that this is further equal to
\begin{align*}
& \frac{1}{\sqrt{|2\pi s^2\Xi^j|}}e^{\frac12(\mu^j)^T(W^j)^{-1}\mu^j}
\int e^{-\frac12(\theta^j - (W^j)^{-1}\mu^j)^TW^j (\theta^j -(W^j)^{-1}\mu^j)}d \th^j\\
& = \frac{1}{\sqrt{|2\pi s^2\Xi^j|}}e^{\frac12(\mu^j)^T(W^j)^{-1}\mu^j} \sqrt{|2\pi (W^j)^{-1}|}.
\end{align*}
This completes the proof.
\end{proof}

 As a consequence of Lemma \ref{lem:predictive_density},  we have
\begin{equation}\label{eq:ratio_pred}
   2 \log {B}({j'} \mid j) = (\mu^{j'})^T (W^{j'})^{-1} \mu^{j'}-(\mu^{j})^T (W^{j})^{-1} \mu^{j}
   + \log \Big(  \frac{|s^2W^{j}\Xi^{j}|}{|s^2 W^{j'}\Xi^{j'}|}\Big).
\end{equation}
We now show how the right-hand-side of the display can be evaluated in a numerical efficient and stable way. In the context of Gaussian Markov random fields related tricks have been used in \cite{Rue}.

Suppose  ${j'-j} = {k} > 0$ (if $k = 0$, ${B}(j'\mid j) = 1$ and for $k < 0$ the calculations are similar.)
First we compute $\mu^{j+k}$ and the Cholesky decomposition of  $W^{j+k}$ (the matrix is symmetric and positive definite, so its Cholesky decomposition exists). We obtain an upper triangular matrix $M^{j+k}$
such that $$W^{j+k} = (M^{j+k})^T M^{j+k}.$$
Next we apply the following theorem, taken from \cite{StewartI} (cf.\ Theorem 1.6 therein).

\begin{thm}\label{thm:blockLU}
Suppose the matrix $A$ can be portioned as
$$	A= \begin{bmatrix} A_{11} & A_{12} \\ A_{21} & A_{22}\end{bmatrix},$$
where $A_{11}$ is nonsingular. Then $A$ has a block $LU$ decomposition
$$ \begin{bmatrix} A_{11} & A_{12} \\ A_{21} & A_{22}\end{bmatrix}= \begin{bmatrix} L_{11} & 0 \\ L_{21} & L_{22}\end{bmatrix} \begin{bmatrix} U_{11} & U_{12} \\ 0 & U_{22}\end{bmatrix}
$$
where $L_{11}$ and $U_{11}$ are nonsingular. For such decomposition $A_{11}=L_{11} U_{11}$.
\end{thm}
The Cholesky decomposition factor $M^j$  of $W^j$ hence equals  the upper left block of $M^{j+k}$, which is obtained by retaining only the
first $m_j$ rows and columns of $M_{j+k}$.
Also note that the vector $\mu^j$ is obtained from $\mu^{j+k}$ by retaining only the first $m_j$ elements.

Now if $z^{j+k}$ is the solution to $(M^{j+k})^T z^{j+k}= \mu^{j+k}$, then
$(\mu^{j+k})^T(W^{j+k})^{-1} \mu^{j+k}=\|z^{j+k}\|^2$. If we similarly define $z^j$ as the solution to $(M^j)^T z^j=\mu^j$,
then $z^{j+k} = [z^j, g^{j+k}]$ where $g^{j+k}$ contains the last  $m_{j+k}-m_j$ elements of $z^{j+k}$. Therefore
$$
(\mu^{j+k})^T (W^{j+k})^{-1} \mu^{j+k}-(\mu^{j})^T (W^{j})^{-1} \mu^{j} = \|z^{j+k}\|^2-\|z^j\|^2 = \|g^{j+k}\|^2.
$$
Furthermore,
$$	 \log \Big( \frac{|s^2 W^{j}\Xi^{j}|}{|s^2 W^{j+k}\Xi^{j+k}|}\Big)
=-\sum_{i=m_j+1}^{m_j+k}  \log (s^2 \xi^2_{i}) +  \log \Big(\frac{|W^j|}{|W^{j+k}|}\Big).$$
The second term on the right-hand side equals
$$
2\log \Big(\frac{|M^j|}{|M^{j+k}|}\Big) = -  2\sum_{i=m_j+1}^{m_{j+k}} \log M^{j+k}_{i,i}.
$$
Therefore, we have
\begin{equation}\label{eq:bf_numerical}
2 \log {B}({j'} \given j) = \|g^{j+k} \|^2-  2
\sum_{i=m_j+1}^{m_{j+k}} \log \left(s\xi_iM^{j+k}_{i,i}\right).
\end{equation}
We can summarize our findings as follows.

\bigskip

\begin{center}
\begin{tabular}{|p{0.9\textwidth}|}
\hline
{\bf Algorithm to compute the Bayes factor ${B}({j'} \given  j)$ for ${j'}=j+k$}\\
\hline
\hline
$\bullet$  Compute $\mu^{j+k}$, $\Sigma^{j+k}$ and $W^{j+k}$.\\
$\bullet$ Obtain the Cholesky decomposition of $W^{j+k}$ so that $W^{j+k}=(M^{j+k})^TM^{j+k}$ \\
\quad for an upper-triangular matrix $M^{j+k}$.\\
$\bullet$ Solve $z^{j+k}$ from $(M^{j+k})^T z^{j+k} = \mu^{j+k}$ and partition the solution into \\
\quad $z^{j+k}=[z^j, g^{j+k}]$,  where $\dim(z^j)=m_j$.\\
$\bullet$ Compute ${B}({j'} \given j)$ from (\ref{eq:bf_numerical}). \\
\hline
\end{tabular}
\end{center}

\bigskip

 \section{Concluding remarks}
 Estimation of diffusion processes has attracted a lot of attention in the past two decades. Within the Bayesian setup very few articles have considered the problem of nonparametric estimation. In this article we propose an alternative approach to the method detailed in  \cite{Pokern}. From the simulations it turns out that our method can provide good results.

The simulation results indicate that the posterior mean can be off the truth if  the prior specification is inappropriate in the sense that
\begin{itemize}
\item  the multiplicative scale $s$ is fixed  at a value either too high or too low;
\item  a truncation level is fixed  and  the smoothness of the prior (governed by  $\beta$) is chosen  inappropriately.
\end{itemize}
The first of these problems can be circumvented by specifying a prior distribution on the scaling parameter. As regards the second problem, endowing the truncation level with a prior and employing a reversible jump algorithm, it turns out that reasonable results can be obtained if we erroneously undersmooth by choosing the regularity of the prior too small. For a fixed high truncation level this is certainly not the case. In case the prior is smoother than the true drift function, both reversible jumps and a  fixed high-level model
can  give bad results. Overall however,  simulation results indicate that our method is  more robust against prior misspecification.

 It will be of great interest to complement our numerical results
with mathematical  results providing theoretical performance guarantees and giving further insight in limitations as well.
Another interesting possible extension is to endow the regularity parameter $\beta$ with a prior as well and let the data
determine its appropriate value. This destroys the partial conjugacy however and it is a challenge to devise numerically feasible
procedures for this approach.

\section*{Acknowledgement}

The authors thank Yvo Pokern for providing the data and code used in the paper \cite{Pokern}.

\bibliographystyle{elsarticle-harv}
\bibliography{lit}

\def\cprime{$'$}
\begin{thebibliography}{24}
\expandafter\ifx\csname natexlab\endcsname\relax\def\natexlab#1{#1}\fi
\providecommand{\url}[1]{\texttt{#1}}
\providecommand{\href}[2]{#2}
\providecommand{\path}[1]{#1}
\providecommand{\DOIprefix}{doi:}
\providecommand{\ArXivprefix}{arXiv:}
\providecommand{\URLprefix}{URL: }
\providecommand{\Pubmedprefix}{pmid:}
\providecommand{\doi}[1]{\href{http://dx.doi.org/#1}{\path{#1}}}
\providecommand{\Pubmed}[1]{\href{pmid:#1}{\path{#1}}}
\providecommand{\bibinfo}[2]{#2}
\ifx\xfnm\relax \def\xfnm[#1]{\unskip,\space#1}\fi
\bibitem[{Beskos et~al.(2006a)Beskos, Papaspiliopoulos and
  Roberts}]{BeskosPapaspiliopoulosRoberts}
\bibinfo{author}{Beskos, A.}, \bibinfo{author}{Papaspiliopoulos, O.},
  \bibinfo{author}{Roberts, G.O.}, \bibinfo{year}{2006}a.
\newblock \bibinfo{title}{Retrospective exact simulation of diffusion sample
  paths with applications}.
\newblock \bibinfo{journal}{Bernoulli} \bibinfo{volume}{12},
  \bibinfo{pages}{1077--1098}.
\newblock \URLprefix \url{http://dx.doi.org/10.3150/bj/1165269151},
  \DOIprefix\doi{10.3150/bj/1165269151}.
\bibitem[{Beskos et~al.(2006b)Beskos, Papaspiliopoulos, Roberts and
  Fearnhead}]{BeskosPapaspiliopoulosRobertsFearnhead}
\bibinfo{author}{Beskos, A.}, \bibinfo{author}{Papaspiliopoulos, O.},
  \bibinfo{author}{Roberts, G.O.}, \bibinfo{author}{Fearnhead, P.},
  \bibinfo{year}{2006}b.
\newblock \bibinfo{title}{Exact and computationally efficient likelihood-based
  estimation for discretely observed diffusion processes}.
\newblock \bibinfo{journal}{J. R. Stat. Soc. Ser. B Stat. Methodol.}
  \bibinfo{volume}{68}, \bibinfo{pages}{333--382}.
\newblock \URLprefix \url{http://dx.doi.org/10.1111/j.1467-9868.2006.00552.x},
  \DOIprefix\doi{10.1111/j.1467-9868.2006.00552.x}. \bibinfo{note}{with
  discussions and a reply by the authors}.
\bibitem[{Brooks et~al.(2011)Brooks, Gelman, Jones and Meng}]{ChapmanHall}
\bibinfo{editor}{Brooks, S.}, \bibinfo{editor}{Gelman, A.},
  \bibinfo{editor}{Jones, G.L.}, \bibinfo{editor}{Meng, X.L.} (Eds.),
  \bibinfo{year}{2011}.
\newblock \bibinfo{title}{Handbook of {M}arkov chain {M}onte {C}arlo}.
\newblock Chapman \& Hall/CRC Handbooks of Modern Statistical Methods,
  \bibinfo{publisher}{CRC Press}, \bibinfo{address}{Boca Raton, FL}.
\newblock \URLprefix \url{http://dx.doi.org/10.1201/b10905},
  \DOIprefix\doi{10.1201/b10905}.
\bibitem[{Eraker(2001)}]{Eraker}
\bibinfo{author}{Eraker, B.}, \bibinfo{year}{2001}.
\newblock \bibinfo{title}{M{CMC} analysis of diffusion models with application
  to finance}.
\newblock \bibinfo{journal}{J. Bus. Econom. Statist.} \bibinfo{volume}{19},
  \bibinfo{pages}{177--191}.
\newblock \URLprefix \url{http://dx.doi.org/10.1198/073500101316970403},
  \DOIprefix\doi{10.1198/073500101316970403}.
\bibitem[{Godsill(2001)}]{Godsill}
\bibinfo{author}{Godsill, S.J.}, \bibinfo{year}{2001}.
\newblock \bibinfo{title}{On the relationship between {M}arkov chain {M}onte
  {C}arlo methods for model uncertainty}.
\newblock \bibinfo{journal}{J. Comput. Graph. Statist.} \bibinfo{volume}{10},
  \bibinfo{pages}{230--248}.
\newblock \URLprefix \url{http://dx.doi.org/10.1198/10618600152627924},
  \DOIprefix\doi{10.1198/10618600152627924}.
\bibitem[{Green(1995)}]{Green}
\bibinfo{author}{Green, P.J.}, \bibinfo{year}{1995}.
\newblock \bibinfo{title}{Reversible jump {M}arkov chain {M}onte {C}arlo
  computation and {B}ayesian model determination}.
\newblock \bibinfo{journal}{Biometrika} \bibinfo{volume}{82},
  \bibinfo{pages}{711--732}.
\newblock \URLprefix \url{http://dx.doi.org/10.1093/biomet/82.4.711},
  \DOIprefix\doi{10.1093/biomet/82.4.711}.
\bibitem[{Green(2003)}]{Green2003}
\bibinfo{author}{Green, P.J.}, \bibinfo{year}{2003}.
\newblock \bibinfo{title}{Trans-dimensional {M}arkov chain {M}onte {C}arlo},
  in: \bibinfo{booktitle}{Highly structured stochastic systems}.
  \bibinfo{publisher}{Oxford Univ. Press}, \bibinfo{address}{Oxford}.
  volume~\bibinfo{volume}{27} of \textit{\bibinfo{series}{Oxford Statist. Sci.
  Ser.}}, pp. \bibinfo{pages}{179--206}.
\newblock \bibinfo{note}{With part A by Simon J. Godsill and part B by Juha
  Heikkinen}.
\bibitem[{Hindriks(2011)}]{Hindriks}
\bibinfo{author}{Hindriks, R.}, \bibinfo{year}{2011}.
\newblock \bibinfo{title}{Empirical dynamics of neuronal rhythms}.
\newblock \bibinfo{note}{PhD thesis, VU University Amsterdam}.
\bibitem[{Kashin and Saakyan(1989)}]{kashin}
\bibinfo{author}{Kashin, B.S.}, \bibinfo{author}{Saakyan, A.A.},
  \bibinfo{year}{1989}.
\newblock \bibinfo{title}{Orthogonal series}. volume~\bibinfo{volume}{75} of
  \textit{\bibinfo{series}{Translations of Mathematical Monographs}}.
\newblock \bibinfo{publisher}{American Mathematical Society},
  \bibinfo{address}{Providence, RI}.
\newblock \bibinfo{note}{Translated from the Russian by Ralph P. Boas,
  Translation edited by Ben Silver}.
\bibitem[{Liptser and Shiryaev(2001)}]{LiptserShiryayevI}
\bibinfo{author}{Liptser, R.S.}, \bibinfo{author}{Shiryaev, A.N.},
  \bibinfo{year}{2001}.
\newblock \bibinfo{title}{Statistics of random processes. {I}}.
  volume~\bibinfo{volume}{5} of \textit{\bibinfo{series}{Applications of
  Mathematics (New York)}}.
\newblock \bibinfo{edition}{expanded} ed.,
  \bibinfo{publisher}{Springer-Verlag}, \bibinfo{address}{Berlin}.
\newblock \bibinfo{note}{General theory, Translated from the 1974 Russian
  original by A. B. Aries, Stochastic Modelling and Applied Probability}.
\bibitem[{Van~der Meulen and van Zanten(2013)}]{VdMeulenVZanten}
\bibinfo{author}{Van~der Meulen, F.H.}, \bibinfo{author}{van Zanten, J.H.},
  \bibinfo{year}{2013}.
\newblock \bibinfo{title}{Consistent nonparametric {B}ayesian inference for
  discretely observed scalar diffusions}.
\newblock \bibinfo{journal}{Bernoulli} \bibinfo{volume}{19},
  \bibinfo{pages}{44--63}.
\bibitem[{Panzar and van Zanten(2009)}]{PanzarVZanten}
\bibinfo{author}{Panzar, L.}, \bibinfo{author}{van Zanten, J.H.},
  \bibinfo{year}{2009}.
\newblock \bibinfo{title}{Nonparametric {B}ayesian inference for ergodic
  diffusions}.
\newblock \bibinfo{journal}{J. Statist. Plann. Inference}
  \bibinfo{volume}{139}, \bibinfo{pages}{4193--4199}.
\newblock \URLprefix \url{http://dx.doi.org/10.1016/j.jspi.2009.06.003},
  \DOIprefix\doi{10.1016/j.jspi.2009.06.003}.
\bibitem[{Papaspiliopoulos et~al.(2012)Papaspiliopoulos, Pokern, Roberts and
  Stuart}]{Pokern}
\bibinfo{author}{Papaspiliopoulos, O.}, \bibinfo{author}{Pokern, Y.},
  \bibinfo{author}{Roberts, G.O.}, \bibinfo{author}{Stuart, A.},
  \bibinfo{year}{2012}.
\newblock \bibinfo{title}{Nonparametric estimation of diffusions: a
  differential equations approach}.
\newblock \bibinfo{journal}{Biometrika} \bibinfo{volume}{99},
  \bibinfo{pages}{511--531}.
\bibitem[{Pokern(2007)}]{Yvo}
\bibinfo{author}{Pokern, Y.}, \bibinfo{year}{2007}.
\newblock \bibinfo{title}{Fitting Stochastic Differential Equations to
  Molecular Dynamics Data}.
\newblock \bibinfo{note}{PhD thesis, University of Warwick}.
\bibitem[{Pokern et~al.(2013)Pokern, Stuart and van
  Zanten}]{PokernStuartvanZanten}
\bibinfo{author}{Pokern, Y.}, \bibinfo{author}{Stuart, A.M.},
  \bibinfo{author}{van Zanten, J.H.}, \bibinfo{year}{2013}.
\newblock \bibinfo{title}{Posterior consistency via precision operators for
  {B}ayesian nonparametric drift estimation in {SDE}s.}
\newblock \bibinfo{journal}{Stoch. Proc. Appl.} \bibinfo{volume}{123},
  \bibinfo{pages}{603--628}.
\bibitem[{Robert(2007)}]{Robert}
\bibinfo{author}{Robert, C.}, \bibinfo{year}{2007}.
\newblock \bibinfo{title}{The Bayesian Choice: From Decision-Theoretic
  Foundations to Computational Implementation}.
\newblock Springer Texts in Statistics, \bibinfo{publisher}{Springer}.
\newblock \URLprefix \url{http://books.google.com/books?id=6oQ4s8Pq9pYC}.
\bibitem[{Roberts and Stramer(2001)}]{RobertsStramer}
\bibinfo{author}{Roberts, G.O.}, \bibinfo{author}{Stramer, O.},
  \bibinfo{year}{2001}.
\newblock \bibinfo{title}{On inference for partially observed nonlinear
  diffusion models using the {M}etropolis-{H}astings algorithm}.
\newblock \bibinfo{journal}{Biometrika} \bibinfo{volume}{88},
  \bibinfo{pages}{603--621}.
\newblock \URLprefix \url{http://dx.doi.org/10.1093/biomet/88.3.603},
  \DOIprefix\doi{10.1093/biomet/88.3.603}.
\bibitem[{Rogers and Williams(2000)}]{RW}
\bibinfo{author}{Rogers, L.C.G.}, \bibinfo{author}{Williams, D.},
  \bibinfo{year}{2000}.
\newblock \bibinfo{title}{Diffusion, Markov Processes and Martingales, {V}olume
  1}.
\newblock \bibinfo{edition}{2nd} ed., \bibinfo{publisher}{Cambridge}.
\bibitem[{Rose(1970)}]{Rose}
\bibinfo{author}{Rose, D.J.}, \bibinfo{year}{1970}.
\newblock \bibinfo{title}{Triangulated graphs and the elimination process}.
\newblock \bibinfo{journal}{J. Math. Anal. Appl.} \bibinfo{volume}{32},
  \bibinfo{pages}{597--609}.
\bibitem[{Rue et~al.(2009)Rue, Martino and Chopin}]{Rue}
\bibinfo{author}{Rue, H.}, \bibinfo{author}{Martino, S.},
  \bibinfo{author}{Chopin, N.}, \bibinfo{year}{2009}.
\newblock \bibinfo{title}{Approximate {B}ayesian inference for latent
  {G}aussian models by using integrated nested {L}aplace approximations}.
\newblock \bibinfo{journal}{J. R. Stat. Soc. Ser. B Stat. Methodol.}
  \bibinfo{volume}{71}, \bibinfo{pages}{319--392}.
\newblock \URLprefix \url{http://dx.doi.org/10.1111/j.1467-9868.2008.00700.x},
  \DOIprefix\doi{10.1111/j.1467-9868.2008.00700.x}.
\bibitem[{Shen and Wasserman(2001)}]{ShenWasserman}
\bibinfo{author}{Shen, X.}, \bibinfo{author}{Wasserman, L.},
  \bibinfo{year}{2001}.
\newblock \bibinfo{title}{Rates of convergence of posterior distributions}.
\newblock \bibinfo{journal}{Ann. Statist.} \bibinfo{volume}{29},
  \bibinfo{pages}{687--714}.
\newblock \URLprefix \url{http://dx.doi.org/10.1214/aos/1009210686},
  \DOIprefix\doi{10.1214/aos/1009210686}.
\bibitem[{Stewart(1998)}]{StewartI}
\bibinfo{author}{Stewart, G.W.}, \bibinfo{year}{1998}.
\newblock \bibinfo{title}{Matrix algorithms. {V}ol. {I}}.
\newblock \bibinfo{publisher}{Society for Industrial and Applied Mathematics},
  \bibinfo{address}{Philadelphia, PA}.
\newblock \bibinfo{note}{Basic decompositions}.
\bibitem[{Tanner and Wong(1987)}]{TannerWong}
\bibinfo{author}{Tanner, M.A.}, \bibinfo{author}{Wong, W.H.},
  \bibinfo{year}{1987}.
\newblock \bibinfo{title}{The calculation of posterior distributions by data
  augmentation}.
\newblock \bibinfo{journal}{J. Amer. Statist. Assoc.} \bibinfo{volume}{82},
  \bibinfo{pages}{528--550}.
\newblock \URLprefix
  \url{http://links.jstor.org/sici?sici=0162-1459(198706)82:398<528:TCOPDB>2.0%
.CO;2-M&origin=MSN}. \bibinfo{note}{with discussion and with a reply by the
  authors}.
\bibitem[{Zhao(2000)}]{Zhao}
\bibinfo{author}{Zhao, L.H.}, \bibinfo{year}{2000}.
\newblock \bibinfo{title}{Bayesian aspects of some nonparametric problems}.
\newblock \bibinfo{journal}{Ann. Statist.} \bibinfo{volume}{28},
  \bibinfo{pages}{532--552}.
\newblock \URLprefix \url{http://dx.doi.org/10.1214/aos/1016218229},
  \DOIprefix\doi{10.1214/aos/1016218229}.

\end{thebibliography}






\end{document}